\documentclass[a4paper]{article}
 
\usepackage[margin=1in]{geometry}
\usepackage{enumerate}
\usepackage{amssymb}
\usepackage{amsthm}
\usepackage{amsmath}
\usepackage{graphicx}
\usepackage{stmaryrd}
\usepackage[algoruled,linesnumbered,algo2e,vlined]{algorithm2e}
\SetArgSty{textrm}
\usepackage[final,mode=multiuser]{fixme}
\fxuseenvlayout{colorsig}
\fxusetargetlayout{color}
\FXRegisterAuthor{tn}{atn}{TN}
\FXRegisterAuthor{ti}{ati}{TI}
\FXRegisterAuthor{si}{asi}{SI}
\FXRegisterAuthor{hb}{ahb}{HB}
\FXRegisterAuthor{mt}{amt}{MT}
\definecolor{kgnote}{rgb}{1.0000,0.0000,0.0000}
\fxsetface{env}{}
\newcommand{\derive}{\mathit{val}}

\newlength\savedwidth

\newcommand{\rev}[1]{#1^{R}}

\newcommand{\LCPQ}{\mathsf{LCP}}
\newcommand{\LCSQ}{\mathsf{LCS}}
\newcommand{\LCEQ}{\mathsf{LCE}}

\newcommand{\shrink}[2]{\mathit{Shrink}_{#1}^{#2}}
\newcommand{\xshrink}[2]{\mathit{XShrink}_{#1}^{#2}}
\newcommand{\pow}[2]{\mathit{Pow}_{#1}^{#2}}
\newcommand{\xpow}[2]{\mathit{XPow}_{#1}^{#2}}
\newcommand{\uniq}[1]{\mathit{Uniq}(#1)}

\newcommand{\id}[1]{\mathit{id}(#1)}

\newcommand{\encblock}[1]{\mathit{Eblock}(#1)}

\newcommand{\encblockd}[1]{\mathit{Eblock}_{d}(#1)}
\newcommand{\encpow}[1]{\mathit{Epow}(#1)}

\newcommand{\deltaLR}[1]{\Delta_{#1}}

\newcommand{\sig}[1]{\mathit{Sig}(#1)}

\newcommand{\val}[1]{\mathit{val}(#1)}
\newcommand{\valp}[1]{\mathit{val}^{+}(#1)}

\newcommand{\lcp}[2]{\mathsf{LCP}(#1,#2)}
\newcommand{\lcs}[2]{\mathsf{LCS}(#1,#2)}

\newtheorem{theorem}{Theorem}
\newtheorem{definition}[theorem]{Definition}
\newtheorem{lemma}[theorem]{Lemma}

\newtheorem{example}[theorem]{Example}

\title{Fully dynamic data structure for LCE queries in compressed space}

 \author{
 	Takaaki Nishimoto$^1$\quad
 	Tomohiro I$^2$\quad
 	Shunsuke Inenaga$^1$\\
 	Hideo Bannai$^1$\quad
 	Masayuki Takeda$^1$\\
 	{$^1$ Department of Informatics, Kyushu University}\\
 	{\texttt{\{takaaki.nishimoto,inenaga,bannai,takeda\}@inf.kyushu-u.ac.jp}}\\
 	{$^2$ Kyushu Institute of Technology, Japan}\\
 	{\texttt{tomohiro@ai.kyutech.ac.jp}}
 }

\date{}

\begin{document}

\maketitle

\begin{abstract}
A Longest Common Extension (LCE) query on a text $T$ of length $N$ asks 
for the length of the longest common prefix of suffixes starting at given two positions.
We show that the signature encoding $\mathcal{G}$ of size $w = O(\min(z \log N \log^* M, N))$ 
[Mehlhorn et al., Algorithmica 17(2):183-198, 1997] of $T$,
which can be seen as a compressed representation of $T$,
has a capability to support LCE queries in $O(\log N + \log \ell \log^* M)$ time, where
$\ell$ is the answer to the query,
$z$ is the size of the Lempel-Ziv77 (LZ77) factorization of $T$, and
$M \geq 4N$ is an integer that can be handled in constant time under word RAM model.
In compressed space, this is the fastest deterministic LCE data structure in many cases.
Moreover, $\mathcal{G}$ can be enhanced to support efficient update operations:
After processing $\mathcal{G}$ in $O(w f_{\mathcal{A}})$ time,
we can insert/delete any (sub)string of length $y$ into/from an arbitrary position of $T$ in $O((y+ \log N\log^* M) f_{\mathcal{A}})$ time,
where $f_{\mathcal{A}} = O(\min \{ \frac{\log\log M \log\log w}{\log\log\log M}, \sqrt{\frac{\log w}{\log\log w}} \})$.
This yields the first fully dynamic LCE data structure working in compressed space.
We also present efficient construction algorithms from various types of inputs:
We can construct $\mathcal{G}$ in $O(N f_{\mathcal{A}})$ time from uncompressed string $T$;
in $O(n \log\log (n \log^* M) \log N \log^* M)$ time from grammar-compressed string $T$ represented by a straight-line program of size $n$;
and in $O(z f_{\mathcal{A}} \log N \log^* M)$ time from LZ77-compressed string $T$ with $z$ factors.
On top of the above contributions, we show several applications of our
data structures which improve previous best known results
on grammar-compressed string processing.
\end{abstract}

\section{Introduction}
A \emph{Longest Common Extension (LCE)} query on a text $T$ of length $N$ asks to compute
the length of the longest common prefix of suffixes starting at given two positions.
This fundamental query appears at the heart of many string processing problems (see text book~\cite{gusfield97:_algor_strin_trees_sequen} for example), and hence,
efficient data structures to answer LCE queries gain a great attention.
A classic solution is to use a data structure for lowest common ancestor queries~\cite{DBLP:journals/jal/BenderFPSS05} on the suffix tree of $T$.
Although this achieves constant query time, the $\Theta(N)$ space needed for the data structure is too large to apply it to large scale data.
Hence, recent work focuses on reducing space usage at the expense of query time.
For example, time-space trade-offs of LCE data structure have been extensively studied~\cite{DBLP:conf/cpm/BilleGKLV15, DBLP:journals/corr/TanimuraIBIPT16}.

Another direction to reduce space is to utilize a compressed structure of $T$, which is advantageous when $T$ is highly compressible.
There are several LCE data structures working on grammar-compressed string $T$ represented by a straight-line program (SLP) of size $n$.
The best known deterministic LCE data structure is
due to I et al.~\cite{IMSIBTNS15},
which supports LCE queries in $O(h \log N)$ time,
and occupies $O(n^2)$ space, where $h$ is the height of the derivation tree of a given SLP.
Their data structure can be built in $O(h n^2)$ time directly from the SLP.
Bille et al.~\cite{bille13:_finger_compr_strin} showed 
a Monte Carlo randomized data structure
which supports LCE queries in $O(\log N \log \ell)$ time,
where $\ell$ is the output of the LCE query.
Their data structure requires only $O(n)$ space,
but requires $O(N)$ time to construct.
Very recently, Bille et al.~\cite{BilleCCG15} showed
a faster Monte Carlo randomized data structure of $O(n)$ space
which supports LCE queries in $O(\log N + \log^2 \ell)$ time.
The preprocessing time of this new data structure is not given in~\cite{BilleCCG15}.
Note that, given the LZ77-compression of size $z$ of $T$,
we can convert it into an SLP of size $n = O(z \log \frac{N}{z})$~\cite{rytter03:_applic_lempel_ziv} and then apply the above results.

In this paper, we focus on the \emph{signature encoding} $\mathcal{G}$ of $T$,
which can be seen as a grammar compression of $T$, and show that $\mathcal{G}$ can support LCE queries efficiently.
The signature encoding was proposed by Mehlhorn et al.
for equality testing on a dynamic set of strings~\cite{DBLP:journals/algorithmica/MehlhornSU97}.
Alstrup et al. used signature encodings combined with their own data structure called anchors to present
a pattern matching algorithm on a dynamic set of strings~\cite{DBLP:conf/soda/AlstrupBR00,LongAlstrup}.
In their paper, they also showed that signature encodings can support
longest common prefix (LCP) and longest common suffix (LCS) queries on a dynamic set of strings.
Their algorithm is randomized as it uses a hash table for maintaining the dictionary of $\mathcal{G}$.
Very recently, Gawrychowski et al. improved the results 
by pursuing advantages of randomized approach other than the hash table~\cite{Gawrychowski2015OptimalDynamicStrings}.
It should be noted that the algorithms in~\cite{DBLP:conf/soda/AlstrupBR00,LongAlstrup,Gawrychowski2015OptimalDynamicStrings} can support LCE queries by 
combining split operations and LCP queries although it is not explicitly mentioned.
However, they did not focus on the fact that signature encodings can work in compressed space.
In~\cite{DBLP:conf/latin/0001IK16}, 
LCE data structures on edit sensitive parsing, a variant of signature encoding,
was used for sparse suffix sorting, but again, they did not focus on working in compressed space.

Our contributions are stated by the following theorems, where
$M \geq 4N$ is an integer that can be handled in constant time under word RAM model.
More specifically, $M = 4N$ if $T$ is static, and $M/4$ is the upper bound of the length of $T$ if we consider updating $T$ dynamically.
In dynamic case, $N$ (resp. $w$) always denotes the current size of $T$ (resp. $\mathcal{G}$).
Also, $f_{\mathcal{A}}$ denotes the time for predecessor/successor queries on a set of $w$ integers from an $M$-element universe,
which is $f_{\mathcal{A}} = O(\min \{ \frac{\log\log M \log\log w}{\log\log\log M}, \sqrt{\frac{\log w}{\log\log w}} \})$
by the best known data structure~\cite{DBLP:journals/jcss/BeameF02}.
\begin{theorem}[LCE queries]\label{theo:theorem1}
	Let $\mathcal{G}$ denote the signature encoding of size $w = O(\min(z \log N \\ \log^* M, N))$ for a string $T$ of length $N$. 
	Then $\mathcal{G}$ supports LCE queries on $T$ in $O(\log N + \log \ell \log^* M)$ time,
	where $\ell$ is the answer to the query, and $z$ is the size of the LZ77 factorization of $T$. 
\end{theorem}

\begin{theorem}[Updates]\label{theo:theorem2}
	After processing $\mathcal{G}$ in $O(w f_{\mathcal{A}})$ time,
	we can insert/delete any (sub)string $Y$ of length $y$ into/from an arbitrary position of $T$ in $O((y+ \log N\log^* M) f_{\mathcal{A}})$ time.
	If $Y$ is given as a substring of $T$,
	we can support insertion in $O(f_{\mathcal{A}} \log N \log^* M)$ time.
\end{theorem}

\begin{theorem}[Construction]\label{theo:HConstructuionTheorem}
	Let $T$ be a string of length $N$, 
	$Z$ be LZ77 factorization without self reference of size $z$ representing $T$,
	and $\mathcal{S}$ be an SLP of size $n$ generating $T$.
	Then, we can construct the signature encoding $\mathcal{G}$ for $T$ in
	(1a) in $O(N f_{\mathcal{A}})$ time and $O(w)$ working space from $T$,
	(1b) in $O(N)$ time and working space from $T$,
	(2) in $O(z f_{\mathcal{A}} \log N \log^* M)$ time and $O(w)$ working space from $Z$, 
	(3a) in $O(n f_{\mathcal{A}} \log N \log^* M)$ time and $O(w)$ working space from $\mathcal{S}$, and
	(3b) in $O(n \log \log (n \log^* M) \log N \log^* M)$ time and $O(n \log^* M + w)$ working space from $\mathcal{S}$. 
\end{theorem}

The remarks on our contributions are listed in the following:
\begin{itemize}
  \item We achieve an algorithm for the fastest deterministic LCE queries on SLPs, which even permits faster LCE queries than
        the randomized data structure of Bille et al.~\cite{BilleCCG15} 
        when $\log^* M = o(\log \ell)$ which in many cases is true.
  \item We present the first fully dynamic LCE data structure working in compressed space.
  \item Different from the work in~\cite{DBLP:conf/soda/AlstrupBR00,LongAlstrup,Gawrychowski2015OptimalDynamicStrings},
        we mainly focus on maintaining a single text $T$ in compressed $O(w)$ space.
        For this reason we opt for supporting insertion/deletion as edit operations rather than split/concatenate on a dynamic set of strings.
        However, the difference is not much essential;
        our insert operations specified by a substring of an existing string can work as split/concatenate, and conversely, split/concatenate can simulate insert.
        Our contribution here is to clarify how to collect garbage being produced during edit operations, as directly indicated by a support of delete operations.
  \item The results (2) and (3a) of Theorem~\ref{theo:HConstructuionTheorem} immediately follow from 
        the update operations considered in~\cite{DBLP:conf/soda/AlstrupBR00,LongAlstrup}, but others are nontrivial.
  \item Direct construction of $\mathcal{G}$ from SLPs is important for applications in compressed string processing,
        where the task is to process a given compressed representation of string(s) without explicit decompression.
	In particular, we use the result (3b) of Theorem~\ref{theo:HConstructuionTheorem} to show several applications which improve previous best known results. 
	Note that the time complexity of the result (3b) can be written as $O(n \log \log n \log N \log^* M)$
	when $\log^* M = O(n)$ which in many cases is true, and always true in static case because $\log^* M = O(\log^* N) = O(\log N) = O(n)$.
\end{itemize}

	Proofs and examples omitted due to lack of space are in a full version of this paper~\cite{DBLP:journals/corr/NishimotoIIBT16}.
\section{Preliminaries} \label{sec:preliminary}

\subsection{Strings}

Let $\Sigma$ be an ordered alphabet.
An element of $\Sigma^*$ is called a string.
For string $w = xyz$,
$x$, $y$ and $z$ are called a prefix, substring, and suffix of $w$, respectively.
The length of string $w$ is denoted by $|w|$.
The empty string $\varepsilon$ is a string of length $0$.
Let $\Sigma^+ = \Sigma^* - \{\varepsilon\}$.
For any $1 \leq i \leq |w|$, $w[i]$ denotes the $i$-th character of $w$.
For any $1 \leq i \leq j \leq |w|$,
$w[i..j]$ denotes the substring of $w$ that begins at position $i$
and ends at position $j$.
Let $w[i..] = w[i..|w|]$ and $w[..i] = w[1..i]$ for any $1 \leq i \leq |w|$.
For any string $w$, let $\rev{w}$ denote the reversed string of $w$,
that is, $\rev{w} = w[|w|] \cdots w[2]w[1]$. 
For any strings $w$ and $u$, 
let $\lcp{w}{u}$ (resp. $\lcs{w}{u}$) denote the length of 
the longest common prefix (resp. suffix) of $w$ and $u$.
Given two strings $s_1, s_2$ and two integers $i, j$, let $\LCEQ(s_1, s_2, i, j)$ denote a query which returns $\lcp{s_1[i..|s_1|]}{s_2[j..|s_2|]}$.
Our model of computation is the unit-cost word RAM with machine word 
size of $\Omega(\log_2 M)$ bits,
and space complexities will be evaluated by the number of machine words.
Bit-oriented evaluation of space complexities can be obtained with 
a $\log_2 M$ multiplicative factor.
\begin{definition}[Lempel-Ziv77 factorization~\cite{LZ77}]
The Lempel-Ziv77 (LZ77) factorization of a string $s$ without self-references is 
a sequence $f_1, \ldots, f_z$ of non-empty substrings of $s$ 
such that $s = f_1 \cdots f_z$,
$f_1 = s[1]$, 
and for $1 < i \leq z$, if the character $s[|f_1..f_{i-1}|+1]$ does not occur in $s[|f_1..f_{i-1}|]$, then $f_i = s[|f_1..f_{i-1}|+1]$, otherwise $f_i$ is the longest prefix of $f_i \cdots f_z$ which occurs in $f_1 \cdots f_{i-1}$. 
The size of the LZ77 factorization $f_1, \ldots, f_z$ of string $s$
is the number $z$ of factors in the factorization.
\end{definition}

\subsection{Context free grammars as compressed representation of strings}

\vspace*{0.5pc}
\noindent \textbf{Straight-line programs.}
A \emph{straight-line program} (\emph{SLP}) is a context free grammar in the Chomsky normal form
that generates a single string.
Formally, an SLP that generates $T$ is
a quadruple $\mathcal{G} = (\Sigma, \mathcal{V}, \mathcal{D}, S)$, 
such that
$\Sigma$ is an ordered alphabet of terminal characters;
$\mathcal{V} = \{ X_1, \ldots, X_{n} \}$ is a set of positive integers, called \emph{variables};
$\mathcal{D} = \{X_i \rightarrow \mathit{expr}_i\}_{i = 1}^{n}$ is a set of \emph{deterministic productions} (or \emph{assignments})
with each $\mathit{expr}_i$ being either of form $X_\ell X_r~(1 \leq \ell, r < i)$, or a single character $a \in \Sigma$;
and $S := X_{n} \in \mathcal{V}$ is the start symbol which derives the string $T$.
We also assume that the grammar neither contains \emph{redundant} variables (i.e., there is at most one assignment whose righthand side is $\mathit{expr}$)
nor \emph{useless} variables (i.e., every variable appears at least once in the derivation tree of $\mathcal{G}$).
The \emph{size} of the SLP $\mathcal{G}$ is the number $n$ of productions in $\mathcal{D}$.
In the extreme cases the length $N$ of the string $T$ can be as large as $2^{n-1}$,
however, it is always the case that $n \geq \log_2 N$.

Let $\mathit{val}: \mathcal{V} \rightarrow \Sigma^+$ be the function
which returns the string derived by an input variable.
If $s = \val{X}$ for $X \in \mathcal{V}$,
then we say that the variable $X$ \emph{represents} string $s$.
For any variable sequence $y \in \mathcal{V}^{+}$,
let $\valp{y} = \val{y[1]} \cdots \val{y[|y|]}$.

\vspace*{0.5pc}
\noindent \textbf{Run-length straight-line programs.}
We define \emph{run-length SLPs} (\emph{RLSLPs}), as an extension to SLPs,
which allow \emph{run-length encodings} in the righthand sides of productions, i.e.,
$\mathcal{D}$ might contain a production $X \rightarrow \hat{X}^{k} \in \mathcal{V} \times \mathcal{N}$.
The \emph{size} of the RLSLP is still the number of productions in $\mathcal{D}$
as each production can be encoded in constant space.
Let $\mathit{Assgn}_{\mathcal{G}}$ be the function such that
$\mathit{Assgn}_{\mathcal{G}}(X_i) = \mathit{expr_i}$ iff $X_i \rightarrow \mathit{expr_i} \in \mathcal{D}$.
Also, let $\mathit{Assgn}^{-1}_{\mathcal{G}}$ denote the reverse function of $\mathit{Assgn}_{\mathcal{G}}$.
When clear from the context, 
we write $\mathit{Assgn}_{\mathcal{G}}$ and $\mathit{Assgn}^{-1}_{\mathcal{G}}$
as $\mathit{Assgn}$ and $\mathit{Assgn}^{-1}$, respectively.

\vspace*{0.5pc}
\noindent \textbf{Representation of RLSLPs.}
For an RLSLP $\mathcal{G}$ of size $w$, 
we can consider a DAG of size $w$ as a compact representation of the derivation trees of variables in $\mathcal{G}$.
Each node represents a variable $X$ in $\mathcal{V}$ and store $|\val{X}|$
and out-going edges represent the assignments in $\mathcal{D}$:
For an assignment $X_i \rightarrow X_{\ell}X_r \in \mathcal{D}$,
there exist two out-going edges from $X_i$ to its ordered children $X_{\ell}$ and $X_r$;
and for $X \rightarrow \hat{X}^{k} \in \mathcal{D}$, there is a single edge from $X$ to $\hat{X}$ with the multiplicative factor $k$. 

\section{Signature encoding}\label{sec:Framework}

Here, we recall the \emph{signature encoding} 
first proposed by Mehlhorn et al.~\cite{DBLP:journals/algorithmica/MehlhornSU97}.
Its core technique is \emph{locally consistent parsing} defined as follows:

\begin{lemma}[Locally consistent parsing~\cite{DBLP:journals/algorithmica/MehlhornSU97,LongAlstrup}]\label{lem:CoinTossing}
  Let $W$ be a positive integer.
  There exists a function $f: [0..W]^{\log^* W + 11} \rightarrow \{0,1\}$ such that,
  for any $p \in [1..W]^n$ with $n \geq 2$ and $p[i] \neq p[i+1]$ for any $1 \leq i < n$,
  the bit sequence $d$ defined by 
  $d[i] = f(\tilde{p}[i-\deltaLR{L}], \ldots, \tilde{p}[i+\deltaLR{R}])$ for $1 \leq i \leq n$, satisfies:
  $d[1] = 1$; 
  $d[n] = 0$;
  $d[i] + d[i+1] \leq 1$ for $1 \leq i < n$; 
  and $d[i] + d[i+1] + d[i+2] + d[i+3] \geq 1$ for any $1 \leq i < n-3$;
where $\deltaLR{L} = \log ^*W + 6$, $\deltaLR{R} = 4$, and $\tilde{p}[j] = p[j]$ for all $1 \leq j \leq n$, $\tilde{p}[j] = 0$ otherwise. 
Furthermore, we can compute $d$ in $O(n)$ time using a precomputed table of size $o(\log W)$, which can be computed in $o(\log W)$ time.
\end{lemma}

For the bit sequence $d$ of Lemma~\ref{lem:CoinTossing},
we define the function $\mathit{Eblock}_d(p)$ that 
decomposes an integer sequence $p$ according to $d$:
$\mathit{Eblock}_d(p)$ decomposes $p$ into a sequence
$q_1, \ldots, q_j$ of substrings called \emph{blocks} of $p$,
such that $p = q_1 \cdots q_j$ and 
$q_i$ is in the decomposition iff $d[|q_1 \cdots q_{i-1}|+1] = 1$
for any $1 \leq i \leq j$.
Note that each block is of length from two to four by the property of $d$, i.e.,
$2 \leq |q_i| \leq 4$ for any $1 \leq i \leq j$.
Let $|\mathit{Eblock}_{d}(p)| = j$ and let $\mathit{Eblock}_{d}(s)[i] = q_i$.
We omit $d$ and write $\encblock{p}$ when it is clear from the context, 
and we use implicitly the bit sequence created by Lemma~\ref{lem:CoinTossing} as $d$. 

We complementarily use run-length encoding to get a sequence to which $\mathit{Eblock}$ can be applied.
Formally, for a string $s$, 
let $\encpow{s}$ be the function which 
groups each maximal run of same characters $a$ as $a^k$,
where $k$ is the length of the run. 
$\encpow{s}$ can be computed in $O(|s|)$ time.
Let $|\encpow{s}|$ denote the number of maximal runs of same characters in $s$ and let $\encpow{s}[i]$ denote $i$-th maximal run in $s$.

The signature encoding is the RLSLP $\mathcal{G} = (\Sigma, \mathcal{V}, \mathcal{D}, S)$,
where the assignments in $\mathcal{D}$ are determined by
recursively applying $\mathit{Eblock}$ and $\mathit{Epow}$ to $T$
until a single integer $S$ is obtained.
We call each variable of the signature encoding a \emph{signature},
and use $e$ (for example, $e_i \rightarrow e_{\ell}e_{r} \in \mathcal{D}$)
instead of $X$ to distinguish from general RLSLPs.

For a formal description,
let $E := \Sigma \cup \mathcal{V}^{2} \cup \mathcal{V}^{3} \cup \mathcal{V}^{4} \cup (\mathcal{V} \times \mathcal{N})$
and let $\mathit{Sig}: E \rightarrow \mathcal{V}$ be the function such that:
$\mathit{Sig}(\mathit{x}) = e$ if $(e \rightarrow \mathit{x}) \in \mathcal{D}$;
$\mathit{Sig}(\mathit{x}) = \mathit{Sig}( \mathit{Sig}(\mathit{x}[1..|\mathit{x}|-1]) \mathit{x}[|\mathit{x}|])$
if $\mathit{x} \in \mathcal{V}^{3} \cup \mathcal{V}^{4}$;
or otherwise undefined.
Namely, the function $\mathit{Sig}$ returns, if any,
the lefthand side of the corresponding production of $\mathit{x}$
by recursively applying the $\mathit{Assgn}^{-1}$ function from left to right.
For any $p \in E^*$,
let $\mathit{Sig}^{+}(p) = \sig{p[1]} \cdots \sig{p[|p|]}$.

The signature encoding of string $T$ is defined by the following $\mathit{Shrink}$ and $\mathit{Pow}$ functions:
$\shrink{t}{T} = \mathit{Sig}^{+}(T)$ for $t = 0$, and $\shrink{t}{T} = \mathit{Sig}^{+}(\encblock{\pow{t-1}{T}})$ for $0 < t \leq h$; and
$\pow{t}{T} = \mathit{Sig}^{+}(\encpow{\shrink{t}{T}})$ for $0 \leq t \leq h$;
where $h$ is the minimum integer satisfying $|\pow{h}{T}| = 1$.
Then, the start symbol of the signature encoding is $S = \pow{h}{T}$.
We say that a node is in \emph{level} $t$ in the derivation tree of $S$
if the node is produced by $\shrink{t}{T}$ or $\pow{t}{T}$.
The height of the derivation tree of the signature encoding of $T$ is
$O(h) = O(\log |T|)$.
For any $T \in \Sigma^+$,
let $\id{T} = \pow{h}{T} = S$, i.e.,
the integer $S$ is the signature of $T$.

In this paper, we implement signature encodings by the DAG of RLSLP introduced in Section~\ref{sec:preliminary}.

\section{Compressed LCE data structure using signature encodings}\label{sec:lce}
In this section, we show Theorem~\ref{theo:theorem1}.

\vspace*{0.5pc}
\noindent \textbf{Space requirement of the signature encoding.}
It is clear from the definition of the signature encoding $\mathcal{G}$ of $T$ that the size of $\mathcal{G}$ is less than $4N \leq M$,
and hence, all signatures are in $[1..M-1]$.
Moreover, the next lemma shows that $\mathcal{G}$ requires only \emph{compressed space}:
\begin{lemma}[\cite{18045}]\label{lem:upperbound_signature}
The size $w$ of the signature encoding of $T$ of length $N$ is 
$O(z \log N \log^* M)$,
where $z$ is the number of factors in the LZ77 factorization without self-reference of $T$. 
\end{lemma}
\vspace*{0.5pc}
\noindent \textbf{Common sequences of signatures to all occurrences of same substrings.}
Here, we recall the most important property of the signature encoding,
which ensures the existence of common signatures to all occurrences of same substrings by the following lemma.

\begin{lemma}[common sequences~\cite{18045}]\label{lem:common_sequence2}
	Let $\mathcal{G}$ be a signature encoding for a string $T$. 
	Every substring $P$ in $T$ is represented by a signature sequence $\mathit{Uniq}(P)$ in $\mathcal{G}$ for a string $P$.  
\end{lemma}
$\mathit{Uniq}(P)$, which we call the \emph{common sequence} of $P$, is defined by the following.

\begin{definition}\label{def:xshrink}
For a string $P$, let
\begin{eqnarray*}
 \mathit{XShrink}_t^{P} &=&
  \begin{cases}
   \mathit{Sig}^{+}(P) & \mbox{ for } t = 0, \\
   \mathit{Sig}^{+}(\encblockd{\mathit{XPow}_{t-1}^{P}}[|L_{t}^{P}|..|\mathit{XPow}_{t-1}^{P}|-|R_{t}^{P}|]) & \mbox{ for } 0 < t \leq h^{P}, \\
  \end{cases} \\
\mathit{XPow}_t^{P} &=& \mathit{Sig}^{+}(\encpow{\mathit{XShrink}_t^{P}[|\hat{L}_{t}^{P}| + 1..|\mathit{XShrink}_t^{P}| - |\hat{R}_{t}^{P}]}|) \ \mbox{ for } 0 \leq t < h^{P}, \mbox{ where}
\end{eqnarray*}
\begin{itemize}
  \item $L_{t}^{P}$ is the shortest prefix of $\mathit{XPow}_{t-1}^{P}$ of length at least $\deltaLR{L}$ such that $d[|L_{t}^{P}|+1] = 1$,
  \item $R_{t}^{P}$ is the shortest suffix of $\mathit{XPow}_{t-1}^{P}$ of length at least $\deltaLR{R}+1$ such that $d[|d| - |R_{t}^{P}| + 1] = 1$,
  \item $\hat{L}_{t}^{P}$ is the longest prefix of $\mathit{XShrink}_t^{P}$ such that $|\encpow{\hat{L}_{t}^{P}}|  = 1$,
  \item $\hat{R}_{t}^{P}$ is the longest suffix of $\mathit{XShrink}_t^{P}$ such that $|\encpow{\hat{R}_{t}^{P}}| = 1$, and
  \item $h^{P}$ is the minimum integer such that $|\encpow{\mathit{XShrink}_{h^{P}}^{P}}| \leq \Delta_{L} + \Delta_{R} + 9$.
\end{itemize}
\end{definition}
Note that $\Delta_{L} \leq |L_{t}^{P}| \leq \Delta_{L} + 3$ and $\Delta_{R}+1 \leq |R_{t}^{P}| \leq \Delta_{R} + 4$ hold by the definition. 
Hence $|\xshrink{t+1}{P}| > 0$ holds if $|\encpow{\xshrink{t}{P}}| > \Delta_{L} + \Delta_{R} + 9$. 
Then,
\[ 
\mathit{Uniq}(P) = \hat{L}_{0}^{P}L_{0}^{P} \cdots 
\hat{L}_{h^{P}-1}^{P}L_{h^{P}-1}^{P}\mathit{XShrink}_{h^{P}}^{P}R_{h^{P}-1}^{P}\hat{R}_{h^{P}-1}^{P} \cdots R_{0}^{P}\hat{R}_{0}^{P}.
\] 

We give an intuitive description of Lemma~\ref{lem:common_sequence2}. 
Recall the locally consistent parsing of Lemma~\ref{lem:CoinTossing}. 
Each $i$-th bit of bit sequence $d$ of Lemma~\ref{lem:CoinTossing} for a given string $s$
is determined by $s[i-\deltaLR{L}..i+\deltaLR{R}]$. 
Hence, for two positions $i, j$ such that $P = s[i..i+k-1] = s[j..j+k-1]$ for some $k$, 
$d[i+\deltaLR{L}..i+k-1-\deltaLR{R}] = d[j+\deltaLR{L}..j+k-1-\deltaLR{R}]$ holds, namely, 
``internal'' bit sequences of the same substring of $s$ are equal. 
Since each level of the signature encoding uses the bit sequence,
all occurrences of same substrings in a string share same internal signature sequences,
and this goes up level by level.
$\xshrink{t}{P}$ and $\xpow{t}{P}$ represent signature sequences 
obtained from only internal signature sequences of $\xpow{t-1}{T}$ and $\xshrink{t}{T}$, respectively. 
This means that $\xshrink{t}{P}$ and $\xpow{t}{P}$ are always created over $P$.
From such common signatures we take as short signature sequence as possible for $\mathit{Uniq}(P)$:
Since $\valp{\pow{t-1}{P}} = \valp{L_{t-1}^{P}\xshrink{t}{P}R_{t-1}^{P}}$ and 
$\valp{\shrink{t}{P}} = \valp{\hat{L}_{t}^{P}\xpow{t}{P}\hat{R}_{t}^{P}}$ hold, 
$|\encpow{\mathit{Uniq}(P)}| = O(\log |P| \log^* M)$ and $\valp{\mathit{Uniq}(P)} = P$ hold.
Hence Lemma~\ref{lem:common_sequence2} holds~\footnote{
	The common sequences are conceptually equivalent to
	the \emph{cores}~\cite{maruyama13:_esp} which are defined for the
	\emph{edit sensitive parsing} of a text,
	a kind of locally consistent parsing of the text.
}.

The number of ancestors of nodes corresponding to $\mathit{Uniq}(P)$ is upper bounded by:
\begin{lemma}\label{lem:ancestors}
  Let $\mathcal{G} = (\Sigma, \mathcal{V}, \mathcal{D}, S)$ be a signature encoding for a string $T$, 
  $P$ be a string, and let $\mathcal{T}$ be the derivation tree of a signature $e \in \mathcal{V}$. 
  Consider an occurrence of $P$ in $s$,
  and the induced subtree $X$ of $\mathcal{T}$
  whose root is the root of $\mathcal{T}$ and whose leaves are 
  the parents of the nodes representing $\uniq{P}$, where $s = \val{e}$.
  Then $X$ contains $O(\log^* M)$ nodes for every level and
  $O(\log |s| + \log |P| \log^* M)$ nodes in total.
\end{lemma}

\vspace*{0.5pc}
\noindent \textbf{LCE queries.}
In the next lemma, we show a more general result than Theorem~\ref{theo:theorem1},
which states that the signature encoding supports (both forward and backward) LCE queries on a given arbitrary pair of signatures. 
Theorem~\ref{theo:theorem1} immediately follows from Lemma~\ref{lem:sub_operation_lemma}.
\begin{lemma}\label{lem:sub_operation_lemma}
Using a signature encoding $\mathcal{G} = (\Sigma, \mathcal{V}, \mathcal{D}, S)$ for a string $T$, 
we can support queries $\LCEQ(s_1, s_2, i, j)$ and $\LCEQ(s_1^{R}, s_2^{R}, i, j)$ 
in $O(\log |s_1| + \log |s_2| + \log\ell \log^* M)$ time 
for given two signatures $e_1, e_2 \in \mathcal{V}$ and 
two integers $1 \leq i \leq |s_1|$, $1 \leq j \leq |s_2|$, 
where $s_1 = \val{e_1}$, $s_2 = \val{e_2}$ and $\ell$ is the answer to the $\LCEQ$ query. 
\end{lemma}
\begin{proof}
 We focus on $\LCEQ(s_1, s_2, i, j)$ as $\LCEQ(s_1^{R}, s_2^{R}, i, j)$ is supported similarly.

 Let $P$ denote the longest common prefix of $s_1[i..]$ and $s_2[j..]$. 
 Our algorithm simultaneously traverses two derivation trees rooted at $e_1$ and $e_2$
 and computes $P$ by matching the common signatures greedily from left to right. 
 Recall that $s_1$ and $s_2$ are substrings of $T$. 
 	Since the both substrings $P$ occurring at position $i$ in $\val{e_1}$ and at position $j$ in $\val{e_2}$ are 
 	represented by $\uniq{P}$ in the signature encoding by Lemma~\ref{lem:common_sequence2},
 we can compute $P$ by at least finding the common sequence of nodes which represents $\uniq{P}$,
 and hence, we only have to traverse ancestors of such nodes.
 By Lemma~\ref{lem:ancestors},
 the number of nodes we traverse, which dominates the time complexity, is upper bounded by
 $O(\log |s_1| + \log |s_2| + \encpow{\uniq{P}}) = O(\log |s_1| + \log |s_2| + \log \ell \log^* M)$.
\end{proof}

\section{Updates}\label{sec:Update}
In this section, we show Theorem~\ref{theo:theorem2}. 
Formally, we consider a \emph{dynamic signature encoding} $\mathcal{G}$ of $T$,
which allows for efficient updates of $\mathcal{G}$ in compressed space according to the following operations:
$\mathit{INSERT}(Y, i)$ inserts a string $Y$ into $T$ at position $i$, i.e., $T \leftarrow T[..i-1] Y T[i..]$;
$\mathit{INSERT'}(j, y, i)$ inserts $T[j..j+y-1]$ into $T$ at position $i$, i.e., $T \leftarrow T[..i-1] T[j..j+y-1] T[i..]$; and
$\mathit{DELETE}(j, y)$ deletes a substring of length $y$ starting at $j$, i.e., $T \leftarrow T[..j-1]T[j+y..]$.

During updates we recompute $\shrink{t}{T}$ and $\pow{t}{T}$ for some part of new $T$
(note that the most part is unchanged thanks to the virtue of signature encodings, Lemma~\ref{lem:ancestors}).
When we need a signature for $\mathit{expr}$, 
we look up the signature assigned to $\mathit{expr}$ (i.e., compute $\mathit{Assign}^{-1}(\mathit{expr})$) and use it if such exists.
If $\mathit{Assign}^{-1}(\mathit{expr})$ is undefined
we create a new signature, which is an integer that is currently not used as signatures (say $e_{\mathit{new}} = \min ([1..M] \setminus \mathcal{V})$),
and add $e_{\mathit{new}} \rightarrow \mathit{expr}$ to $\mathcal{D}$.
Also, updates may produce a useless signature whose parents in the DAG are all removed.
We remove such useless signatures from $\mathcal{G}$ during updates. 

Note that the corresponding nodes and edges of the DAG can be added/removed
in constant time per addition/removal of an assignment. 
In addition to the DAG, we need dynamic data structures to conduct the following operations efficiently:
(A) computing $\mathit{Assgn}^{-1}(\cdot)$,
(B) computing $\min ([1..M] \setminus \mathcal{V})$, and
(C) checking if a signature $e$ is useless.

For (A), we use Beame and Fich's data structure~\cite{DBLP:journals/jcss/BeameF02} 
that can support predecessor/successor queries on a dynamic set of integers.\footnote{
	Alstrup et al.~\cite{LongAlstrup} used hashing for this purpose.
	However,
	since we are interested in the worst case time complexities, we use
	the data structure~\cite{DBLP:journals/jcss/BeameF02} in place of hashing.
}
For example, we consider Beame and Fich's data structure maintaining a set of integers
$\{ e_{\ell} M^2 + e_{r} M + e \mid e \rightarrow e_{\ell} e_{r} \in \mathcal{D} \}$ in $O(w)$ space.
Then we can implement $\mathit{Assgn}^{-1}(e_{\ell} e_{r})$ by computing the successor $q$ of $e_{\ell} M^2 + e_{r} M$, i.e.,
$e = q \mod M$ if $\lfloor q / M \rfloor = e_{\ell} M + e_{r}$, and otherwise $\mathit{Assgn}^{-1}(e_{\ell} e_{r})$ is undefined.
Queries as well as update operations can be done in deterministic $O(f_{\mathcal{A}})$ time, where 
$f_{\mathcal{A}} = O\left(\min \left\{ \frac{\log \log M \log \log w}{\log \log \log M}, \sqrt{\frac{\log w}{\log\log w}} \right\} \right)$. 

For (B), we again use Beame and Fich's data structure to maintain the set of maximal intervals
such that every element in the intervals is signature.
Formally, the intervals are maintained by a set of integers
$\{e_i M + e_j \mid [e_i..e_j] \subseteq \mathcal{V}, e_i - 1 \notin \mathcal{V}, e_j + 1 \notin \mathcal{V} \}$ in $O(w)$ space.
Then we can know the minimum integer currently not in $\mathcal{V}$ by computing the successor of $0$.

For (C), we let every signature $e \in \mathcal{V}$ have a counter to count the number of parents of $e$ in the DAG.
Then we can know that a signature is useless if the counter is $0$.

Lemma~\ref{lem:ComputeShortCommonSequence} shows that we can efficiently compute $\uniq{P}$ for a substring $P$ of $T$.
\begin{lemma}\label{lem:ComputeShortCommonSequence}
Using a signature encoding $\mathcal{G} = (\Sigma, \mathcal{V}, \mathcal{D}, S)$ of size $w$,  
given a signature $e \in \mathcal{V}$ (and its corresponding node in the DAG)
and two integers $j$ and $y$, 
we can compute $\encpow{\uniq{s[j..j+y-1]}}$ in $O(\log |s| + \log y \log^* M)$ time, 
where $s = \val{e}$.
\end{lemma}
 
\begin{proof}[Proof of Theorem~\ref{theo:theorem2}]
	It is easy to see that, given the static signature encoding of $T$, we can construct data structures (A)-(C) in $O(w f_{A})$ time.
	After constructing these, we can add/remove an assignment in $O(f_{\mathcal{A}})$ time.

	Let $\mathcal{G} = (\Sigma, \mathcal{V}, \mathcal{D}, S)$ be the signature encoding before the update operation.
	We support $\mathit{DELETE}(j, y)$ as follows:
	(1) Compute the new start variable $S' = \id{T[..j-1]T[j+y..]}$ 
	by recomputing the new signature encoding from $\uniq{T[..j-1]}$ and $\uniq{T[j+y..]}$. 
		Although we need a part of $d$ to recompute $\encblockd{\pow{t}{T[..j-1]T[j+y..]}}$ for every level $t$, 
		the input size to compute the part of $d$ is $O(\log^* M)$ by Lemma~\ref{lem:CoinTossing}. 
		Hence these can be done in $O(f_{\mathcal{A}}\log N \log^* M)$ time 
		by Lemmas~\ref{lem:ComputeShortCommonSequence} and \ref{lem:ancestors}. 
	(2) Remove all useless signatures $Z$ from $\mathcal{G}$.
	Note that if a signature is useless, then all the signatures along the path from $S$ to it are also useless.
	Hence, we can remove all useless signatures efficiently by depth-first search starting from $S$,
	which takes $O(f_{\mathcal{A}}|Z|)$ time, where $|Z| = O(y + \log N \log^* M)$ by Lemma~\ref{lem:ancestors}.
	
	Similarly, we can support $\mathit{INSERT}(Y, i)$ in $O(f_{\mathcal{A}}(y + \log N \log^* M))$ time 
	by creating the new start variable $S'$ from $\uniq{T[..i-1]}$, $\uniq{Y}$ and $\uniq{T[i..]}$.
	Note that we can naively compute $\uniq{Y}$ in $O(f_{\mathcal{A}} y)$ time.
	For $\mathit{INSERT'}(j, y, i)$, we can avoid $O(f_{\mathcal{A}} y)$ time by computing $\uniq{T[j..j+y-1]}$ using Lemma~\ref{lem:ComputeShortCommonSequence}.
\end{proof}

\section{Construction}\label{sec:Construction}
In this section, we give proofs of Theorem~\ref{theo:HConstructuionTheorem}, but
we omit proofs of the results (2) and (3a)
as they are straightforward from the previous work~\cite{DBLP:conf/soda/AlstrupBR00,LongAlstrup}. 
\subsection{Theorem~\ref{theo:HConstructuionTheorem}~(1a)}
\begin{proof}[Proof of Theorem~\ref{theo:HConstructuionTheorem}~(1a)]
	Note that we can naively compute $\id{T}$ for a given string $T$ in 
	$O(N f_{\mathcal{A}})$ time and $O(N)$ working space.
	In order to reduce the working space,
	we consider factorizing $T$ into blocks of size $B$ and processing them incrementally:
	Starting with the empty signature encoding $\mathcal{G}$,
	we can compute $\id{T}$ in $O(\frac{N}{B}f_{\mathcal{A}}(\log N \log^* M + B))$ time 
	and $O(w + B)$ working space
	by using $\mathit{INSERT}(T[(i-1)B+1..iB],(i-1)B+1)$ for $i = 1, \ldots, {\frac{N}{B}}$ in increasing order.
	Hence our proof is finished by choosing $B = \log N \log^* M$.
\end{proof}

\subsection{Theorem~\ref{theo:HConstructuionTheorem}~(1b)}\label{sec:Proof_HConstructuionTheorem1}
We compute signatures level by level, i.e.,
construct $\shrink{0}{T}, \pow{0}{T}$, $\ldots, \shrink{h}{T}, \pow{h}{T}$ incrementally.
For each level, we create signatures by sorting signature blocks (or run-length encoded signatures) to which we give signatures,
as shown by the next two lemmas.

\begin{lemma} \label{lem:signature_encblock}
Given $\encblock{\pow{t-1}{T}}$ for $0 < t \leq h$, 
we can compute $\shrink{t}{T}$ in $O((b-a)+|\pow{t-1}{T} |)$ time and space,
where $b$ is the maximum integer in $\pow{t-1}{T}$ 
and $a$ is the minimum integer in $\pow{t-1}{T}$.
\end{lemma}
\begin{proof}
Since we assign signatures to signature blocks and run-length signatures in the derivation tree of $S$ in the order they appear in the signature encoding.
$\pow{t-1}{T}[i] - a$ fits in an entry of a bucket of size $b-a$ for each element of $\pow{t-1}{T}[i]$ of $\pow{t-1}{T}$.
Also, the length of each block is at most four.
Hence we can sort all the blocks of $\encblock{\pow{t-1}{T}}$ by bucket sort in $O((b-a)+|\pow{t-1}{T}|)$ time and space. 
Since $\mathit{Sig}$ is an injection
and since we process the levels in increasing order,
for any two different levels $0 \leq t' < t \leq h$,
no elements of $\shrink{t-1}{T}$ appear in $\shrink{t'-1}{T}$,
and hence no elements of $\pow{t-1}{T}$ appear in $\pow{t'-1}{T}$.
Thus, we can determine a new signature for each block in $\encblock{\pow{t-1}{T}}$,
\emph{without} searching existing signatures in the lower levels.
This completes the proof.
\end{proof}

\begin{lemma} \label{lem:signature_encpow}
Given $\encpow{\shrink{t}{T}}$, we can compute $\pow{t}{T}$ in 
$O(x + (b-a) +|\encpow{\shrink{t}{T}} \\ |)$ time and space, 
where $x$ is the maximum length of runs in $\encpow{\shrink{t}{T}}$,
$b$ is the maximum integer in $\pow{t-1}{T}$, 
and $a$ is the minimum integer in $\pow{t-1}{T}$.
\end{lemma}
\begin{proof}
We first sort all the elements of $\encpow{\shrink{t}{T}}$ by bucket sort
in $O(b-a + \\|\encpow{\shrink{t}{T}}|)$ time and space, ignoring the powers of runs.
Then, for each integer $r$ appearing in $\shrink{t}{T}$,
we sort the runs of $r$'s by bucket sort with a bucket of size $x$.
This takes a total of $O(x + |\encpow{\shrink{t}{T}}|)$ time and space
for all integers appearing in $\shrink{t}{T}$.
The rest is the same as the proof of Lemma~\ref{lem:signature_encblock}.
\end{proof}

\begin{proof}[Proof of Theorem~\ref{theo:HConstructuionTheorem}~(1b)]
Since the size of the derivation tree of $\id{T}$ is $O(N)$,
by Lemmas~\ref{lem:CoinTossing}, \ref{lem:signature_encblock}, and~\ref{lem:signature_encpow}, 
we can compute a DAG of $\mathcal{G}$ for $T$ in $O(N)$ time and space.
\end{proof}

\subsection{Theorem~\ref{theo:HConstructuionTheorem}~(3b)}\label{sec:HConstruction3-2}
In this section, we sometimes abbreviate $\val{X}$ as $X$ for $X \in \mathcal{S}$.
For example, $\shrink{t}{X}$ and $\pow{t}{X}$ represents $\shrink{t}{\val{X}}$ and $\pow{t}{\val{X}}$ respectively. 

Our algorithm computes signatures level by level, i.e.,
constructs incrementally $\shrink{0}{X_n}, \\ \pow{0}{X_n}$, $\ldots, \shrink{h}{X_n}, \pow{h}{X_n}$.
Like the algorithm described in Section~\ref{sec:Proof_HConstructuionTheorem1}, we can create signatures 
by sorting blocks of signatures or run-length encoded signatures in the same level.
The main difference is that we now utilize the structure of the SLP, 
which allows us to do the task efficiently in $O(n \log^* M + w)$ working space.
In particular, although $|\shrink{t}{X_n}|, |\pow{t}{X_n}| = O(N)$ for $0 \leq t \leq h$, 
they can be represented in $O(n \log^* M)$ space.

In so doing, we introduce some additional notations relating to $\xshrink{t}{P}$ and $\xpow{t}{P}$ in Definition~\ref{def:xshrink}.
By Lemma~\ref{lem:common_sequence2}, 
there exist $\hat{z}^{(P_1,P_2)}_t$ and $z^{(P_1,P_2)}_t$ 
for any string $P = P_1P_2$ such that the following equation holds:
$\xshrink{t}{P} = \hat{y}^{P_1}_{t} \hat{z}^{(P_1,P_2)}_t \hat{y}^{P_2}_{t}$ for $0 < t \leq h^{P}$, and
$\xpow{t}{P} = y^{P_1}_{t} z^{(P_1,P_2)}_t y^{P_2}_{t}$ for $0 \leq t < h^{P}$,
where we define $\hat{y}^{P}_{t}$ and $y^{P}_{t}$ for a string $P$ as:
\begin{eqnarray*}
\hat{y}^{P}_{t} =
  \begin{cases}
  \xshrink{t}{P} &\mbox{ for } 0 <  t \leq h^{P},\\
  \varepsilon &\mbox{ for } t > h^{P},\\
  \end{cases}&
y^{P}_{t} =
  \begin{cases}
  \xpow{t}{P} &\mbox{ for } 0 \leq t < h^{P},\\
  \varepsilon &\mbox{ for } t \geq h^{P}.\\
  \end{cases}
\end{eqnarray*}
For any variable $X_i \rightarrow X_{\ell} X_{r}$,
we denote $\hat{z}^{X_i}_{t} = \hat{z}^{(\val{X_{\ell}},\val{X_{r}})}_{t}$ (for $0 < t \leq h^{\val{X_i}}$)
and $z^{X_i}_{t} = z^{(\val{X_{\ell}},\val{X_{r}})}_{t}$ (for $0 \leq t < h^{\val{X_i}}$).
Note that $|z^{X_i}_{t}|, |\hat{z}^{X_i}_{t}| = O(\log^* M)$ because 
$z^{X_i}_{t}$ is created on $\hat{R}^{X_{\ell}}_{t}\hat{z}^{X_i}_{t}\hat{L}^{X_{r}}_{t}$, 
similarly, $\hat{z}^{X_i}_{t}$ is created on $R^{X_{\ell}}_{t-1}z^{X_i}_{t-1}L^{X_{r}}_{t-1}$. 
We can use $\hat{z}_{t}^{X_1}, \ldots, \hat{z}_{t}^{X_n}$ (resp. $z_{t}^{X_1}, \ldots, z_{t}^{X_n}$) 
as a compressed representation of $\xshrink{t}{X_n}$ (resp. $\xpow{t}{X_n}$) based on the SLP:
Intuitively, $\hat{z}_{t}^{X_n}$ (resp. $z_{t}^{X_n}$) covers the middle part of $\xshrink{t}{X_n}$ (resp. $\xpow{t}{X_n}$) and
the remaining part is recovered by investigating the left/right child recursively (see also Fig.~\ref{fig:LeftRightAccessFact}).
Hence, with the DAG structure of the SLP, $\xshrink{t}{X_n}$ and $\xpow{t}{X_n}$ can 
be represented in $O(n \log^* M)$ space.
\begin{figure}[t]
\begin{center}
  \includegraphics[scale=0.5]{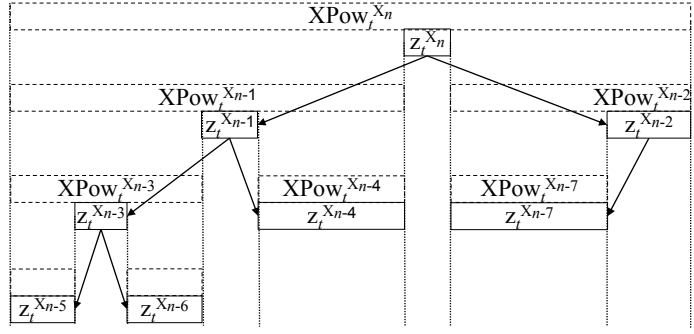}
  \caption{
  $\xpow{t}{X_n}$ can be represented by $z^{X_1}_{t}, \ldots, z^{X_n}_{t}$. 
  In this example, 
  $\xpow{t}{X_n} = z_{t}^{X_{n-5}}z_{t}^{X_{n-3}}z_{t}^{X_{n-6}}z_{t}^{X_{n-1}}z_{t}^{X_{n-4}}z_{t}^{X_{n}}z_{t}^{X_{n-7}}z_{t}^{X_{n-2}}$.
  } 
  \label{fig:LeftRightAccessFact}
\end{center}
\end{figure}

In addition, we define $\hat{A}^{P}_{t}$, $\hat{B}^{P}_{t}$, $A^{P}_t$ and $B^{P}_t$ as follows:
For $0 < t \leq h^{P}$, $\hat{A}^{P}_t$ (resp. $\hat{B}^{P}_t$) is a prefix (resp. suffix) of $\shrink{t}{P}$ 
which consists of signatures of $A^{P}_{t-1}L^{P}_{t-1}$ (resp. $R^{P}_{t-1}B^{P}_{t-1}$); and
for $0 \leq t < h^{P}$, $A^{P}_t$ (resp. $B^{P}_t$) is a prefix (resp. suffix) of $\pow{t}{P}$ 
which consists of signatures of $\hat{A}^{P}_{t}\hat{L}^{P}_{t}$ (resp. $\hat{R}^{P}_{t}\hat{B}^{P}_{t}$).
By the definition,
$\shrink{t}{P} = \hat{A}^{P}_t\xshrink{t}{P}\hat{B}^{P}_t$ for $0 \leq t \leq h^{P}$, and
$\pow{t}{P} = A^{P}_t\xpow{t}{P}B^{P}_t$ for $0 \leq t < h^{P}$.
See Fig.~\ref{fig:AXShrinkB} for the illustration.
\begin{figure}[t]
\begin{center}
  \includegraphics[scale=0.5]{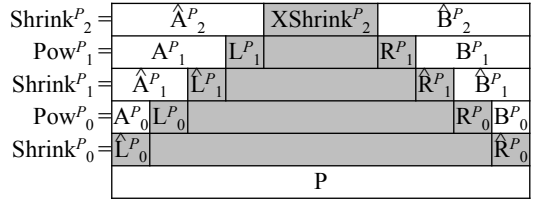}
  \caption{
  An abstract image of $\shrink{t}{P}$ and $\pow{t}{P}$ for a string $P$. 
  For $0 \leq t < h^{P}$, $A^{P}_{t}L^{P}_{t}$ (resp. $R^{P}_{t}B^{P}_{t}$) 
  is encoded into $\hat{A}^{P}_{t+1}$ (resp. $\hat{B}^{P}_{t+1}$). 
  Similarly, for $0 < t < h^{P}$, $\hat{A}^{P}_{t}\hat{L}^{P}_{t}$ (resp. $\hat{R}^{P}_{t}\hat{B}^{P}_{t}$) 
  is encoded into $A^{P}_{t}$ (resp. $B^{P}_{t}$). 
  } 
  \label{fig:AXShrinkB}
\end{center}
\end{figure}

Since $\shrink{t}{X_n} = \hat{A}_{t}^{X_n} \xshrink{t}{X_n} \hat{B}_{t}^{X_n}$ for $0 < t \leq h^{X_n}$,
we use $\hat{\Lambda}_{t} = (\hat{z}_{t}^{X_1}, \ldots, \hat{z}_{t}^{X_n}, \hat{A}^{X_n}_{t}, \\ \hat{B}^{X_n}_{t})$ 
as a compressed representation of $\shrink{t}{X_n}$ of size $O(n \log^* M)$.
Similarly, for $0 \leq t < h^{X_n}$,
we use $\Lambda_{t} = (z_{t}^{X_1}, \ldots, z_{t}^{X_n}, A^{X_n}_{t}, B^{X_n}_{t})$ 
as a compressed representation of $\pow{t}{X_n}$ of size $O(n \log^* M)$.

Our algorithm computes incrementally $\Lambda_{0}, \hat{\Lambda}_{1}, \ldots, \hat{\Lambda}_{h^{X_n}}$.
Given $\hat{\Lambda}_{h^{X_n}}$,
we can easily get $\pow{h^{X_n}}{X_n}$ of size $O(\log^* M)$ in $O(n \log^* M)$ time,
and then $\id{\val{X_n}}$ in $O(\log^* M)$ time from $\pow{h^{X_n}}{X_n}$.
Hence, in the following three lemmas, we show how to compute $\Lambda_{0}, \hat{\Lambda}_{1}, \ldots, \hat{\Lambda}_{h^{X_n}}$.

\begin{lemma}\label{lem:Lambda_0}
Given an SLP of size $n$, we can compute $\Lambda_{0}$ in $O(n \log \log (n \log^* M) \log^* M)$ time and $O(n \log^*M)$ space.
\end{lemma}
\begin{proof}
We first compute, for all variables $X_i$,
$\encpow{\xshrink{0}{X_i}}$ if $|\encpow{\xshrink{0}{X_i}}| \leq \Delta_{L} + \Delta_{R} + 9$,
otherwise $\encpow{\hat{L}_{0}^{X_i}}$ and $\encpow{\hat{R}_{0}^{X_i}}$.
The information can be computed in $O(n \log^*M)$ time and space in a bottom-up manner, i.e., by processing variables in increasing order.
For $X_i \rightarrow X_{\ell} X_{r}$, if both $|\encpow{\xshrink{0}{X_{\ell}}}|$ and $|\encpow{\xshrink{0}{X_{r}}}|$ are no greater than $\Delta_{L} + \Delta_{R} + 9$,
we can compute $\encpow{\xshrink{0}{X_i}}$ in $O(\log^* M)$ time by naively concatenating $\encpow{\xshrink{0}{X_{\ell}}}$ and $\encpow{\xshrink{0}{X_{r}}}$.
Otherwise $|\encpow{\xshrink{0}{X_i}}| > \Delta_{L} + \Delta_{R} + 9$ must hold, and 
$\encpow{\hat{L}_{0}^{X_i}}$ and $\encpow{\hat{R}_{0}^{X_i}}$ can be computed in $O(1)$ time from the information for $X_{\ell}$ and $X_{r}$.

The run-length encoded signatures represented by $z_{0}^{X_i}$ can be obtained by using the above information for $X_{\ell}$ and $X_r$ in $O(\log^* M)$ time:
$z_{0}^{X_i}$ is created over run-length encoded signatures
$\encpow{\xshrink{0}{X_{\ell}}}$ (or $\encpow{\hat{R}_{0}^{X_{\ell}}}$) followed by $\encpow{\xshrink{0}{X_r}}$ (or $\encpow{\hat{R}_{0}^{X_r}}$).
Also, by definition $A_{0}^{X_n}$ and $B_{0}^{X_n}$ represents $\encpow{\hat{L}_{0}^{X_n}}$ and $\encpow{\hat{R}_{0}^{X_n}}$, respectively.

Hence, we can compute in $O(n \log^* M)$ time $O(n \log^*M)$ run-length encoded signatures to which we give signatures.
We determine signatures by sorting the run-length encoded signatures as Lemma~\ref{lem:signature_encpow}.
However, in contrast to Lemma~\ref{lem:signature_encpow},
we do not use bucket sort for sorting the powers of runs
because the maximum length of runs could be as large as $N$ and we cannot afford $O(N)$ space for buckets.
Instead, we use the sorting algorithm of Han~\cite{SuperSort} which sorts $x$ integers in $O(x \log\log x)$ time and $O(x)$ space.
Hence, we can compute $\Lambda_{0}$ in $O(n \log \log (n \log^* M) \log^* M)$ time and $O(n \log^*M)$ space.
\end{proof}

\begin{lemma}\label{lem:Lambda_t}
Given $\hat{\Lambda}_{t}$, we can compute $\Lambda_{t}$ in $O(n \log \log (n \log^* M) \log^*M)$ time and $O(n \log^*M)$ space.
\end{lemma}
\begin{proof}
The computation is similar to that of Lemma~\ref{lem:Lambda_0}
except that we also use $\hat{\Lambda}_{t}$.
\end{proof}

\begin{lemma}\label{lem:hat_Lambda_t}
Given $\Lambda_{t}$, we can compute $\hat{\Lambda}_{t+1}$ in $O(n \log^*M)$ time and $O(n \log^*M)$ space.
\end{lemma}
\begin{proof}
In order to compute $\hat{z}_{t+1}^{X_i}$ for a variable $X_i \rightarrow X_{\ell} X_{r}$,
we need a signature sequence on which $\hat{z}_{t+1}^{X_i}$ is created,
as well as its context, i.e., $\Delta_{L}$ signatures to the left and $\Delta_{R}$ to the right.
To be precise, the needed signature sequence is $v_{t}^{X_{\ell}} z_{t}^{X_i} u_{t}^{X_{r}}$,
where $u_{t}^{X_j}$ (resp. $v_{t}^{X_j}$) denotes a prefix (resp. suffix) of $y_{t}^{X_j}$ of length $\Delta_{L} + \Delta_{R} + 4$ for any variable $X_j$
(see also Figure~\ref{fig:hat_z_construction}). 
Also, we need $A_{t} u_{t}^{X_n}$ and $v_{t}^{X_n} B_{t}$ to create $\hat{A}_{t+1}^{X_n}$ and $\hat{B}_{t+1}^{X_n}$, respectively.

Note that by Definition~\ref{def:xshrink}, 
$|z_{t}^{X}| > \Delta_{L} + \Delta_{R} + 9$ 
if $z_{t}^{X} \neq \varepsilon$.
Then, we can compute $u_{t}^{X_i}$ for all variables $X_i$ in $O(n \log^*M)$ time and space
by processing variables in increasing order on the basis of the following fact:
$u_{t}^{X_i} = u_{t}^{X_{\ell}}$ if $z_{t}^{X_{\ell}} \neq \varepsilon$,
otherwise $u_{t}^{X_i}$ is the prefix of $z_{t}^{X_i}$ of length $\Delta_{L} + \Delta_{R} + 4$.
Similarly $v_{t}^{X_i}$ for all variables $X_i$ can be computed in $O(n \log^*M)$ time and space.

Using $u_{t}^{X_i}$ and $v_{t}^{X_i}$ for all variables $X_i$,
we can obtain $O(n \log^*M)$ blocks of signatures to which we give signatures.
We determine signatures by sorting the blocks by bucket sort as in Lemma~\ref{lem:signature_encblock}
in $O(n \log^*M)$ time.
Hence, we can get $\hat{\Lambda}_{t+1}$ in $O(n \log^*M)$ time and space.
\end{proof}

\begin{figure}[t]
\begin{center}
  \includegraphics[scale=0.5]{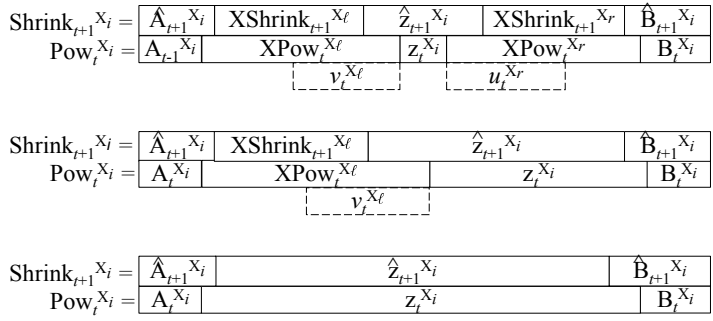}
  \caption{
  Abstract images of the needed signature sequence $v_{t}^{X_{\ell}} z_{t}^{X_i} u_{t}^{X_r}$ ($v_{t}^{X_{\ell}}$ and $u_{t}^{X_r}$ are not shown when they are empty)
  for computing $\hat{z}^{X_i}_{t+1}$ in three situations:
  Top for $0 \leq t < h^{X_{\ell}}, h^{X_{r}}$; middle for $h^{X_r} \leq t < h^{X_{\ell}}$; and bottom for $h^{X_{\ell}}, h^{X_{r}} \leq t < h^{X_i}$.
  } 
  \label{fig:hat_z_construction}
\end{center}
\end{figure}

\begin{proof}[Proof of Theorem~\ref{theo:HConstructuionTheorem}~(3b)]
Using Lemmas~\ref{lem:Lambda_0},~\ref{lem:Lambda_t} and~\ref{lem:hat_Lambda_t},
we can get $\hat{\Lambda}_{h^{X_n}}$ in $O(n \log \log \\(n \log^* M) \log N \log^*M)$ time
by computing $\Lambda_{0}, \hat{\Lambda}_{1}, \ldots, \hat{\Lambda}_{h^{X_n}}$ incrementally.
Note that during the computation we only have to keep $\Lambda_{t}$ (or $\hat{\Lambda}_{t}$) for the current $t$ and the assignments of $\mathcal{G}$.
Hence the working space is $O(n \log^* M + w)$.
By processing $\hat{\Lambda}_{h^{X_n}}$ in $O(n \log^* M)$ time, 
we can get the DAG of $\mathcal{G}$ of size $O(w)$.
\end{proof}

\section{Applications} \label{sec:applications}
Theorem~\ref{theo:changedSLP} is an application to text compression.
Theorems~\ref{theo:SLP_index}-\ref{theo:ImproveDictionaryMatching}
are applications to compressed string processing,
where the task is to process a given compressed representation of string(s)
without explicit decompression.
We believe that only a few applications are listed here,
considering the importance of LCE queries.
As one example of unlisted applications, 
there is a paper~\cite{I2016} in which our LCE data structure was used to improve 
an algorithm of computing the Lyndon factorization of a string represented by a given SLP.
\begin{theorem}\label{theo:changedSLP}
(1) Given a dynamic signature encoding $\mathcal{G}$ for 
  $\mathcal{G} = (\Sigma, \mathcal{V}, \mathcal{D}, S)$ of size $w$ which generates $T$, 
  we can compute an SLP $\mathcal{S}$ of size $O(w \log |T|)$ generating $T$ in $O(w \log |T|)$ time.
(2) Let us conduct a single $\mathit{INSERT}$ or $\mathit{DELETE}$
    operation on the string $T$
    generated by the SLP of (1).
    Let $y$ be the length
    of the substring to be inserted or deleted,
    and let $T'$ be the resulting string.
    During the above operation on the string,
    we can update, in $O((y + \log |T'| \log^* M)(f_{\mathcal{A}} + \log |T'|) )$ time,
    the SLP of (1) to an SLP $\mathcal{S}'$ of size $O(w' \log |T'|)$
    which generates $T'$, 
    where $w'$ is the size of updated $\mathcal{G}$ which generates $T'$.
\end{theorem}

We can get the next lemma using Theorem~\ref{theo:HConstructuionTheorem}~(3b) and Theorem~\ref{theo:theorem2}:
\begin{lemma}\label{lem:SLP_sort}
Given an SLP of size $n$ representing a string of length $N$,
we can sort the variables of the SLP
in lexicographical order in $O(n \log n \log N \log^* N)$ time and $O(n \log^* N + w)$ working space.
\end{lemma}

Lemma~\ref{lem:SLP_sort} has an application to an SLP-based index of Claude and Navarro~\cite{claudear:_self_index_gramm_based_compr}.
In the paper, they showed how to construct their index in $O(n \log n)$ time
if the lexicographic order of variables of a given SLP is already computed.
However, in order to sort variables they almost decompressed the string, and hence,
needs $\Omega(N)$ time and $\Omega(N \log |\Sigma|)$ bits of working space.
Now, Lemma~\ref{lem:SLP_sort} improves the sorting part yielding the next theorem.
\begin{theorem} \label{theo:SLP_index}
Given an SLP of size $n$ representing a string of length $N$,
we can construct the SLP-based index of~\cite{claudear:_self_index_gramm_based_compr}
in $O(n \log n \log N \log^* N)$ time and $O(n \log^* N + w)$ working space.
\end{theorem}

\begin{theorem} \label{theo:faster_LCP}
      Given an SLP $\mathcal{S}$ of size $n$ generating a string $T$ of length $N$,
      we can construct, in $O(n \log \log n \log N \log^* N)$ time, 
      a data structure 
      which occupies $O(n \log N \log^* N)$ space 
      and supports $\LCPQ(\val{X_i}, \val{X_j})$ and $\LCSQ(\val{X_i}, \val{X_j})$ queries 
      for variables $X_i, X_j$ in $O(\log N)$ time.
The $\LCPQ(\val{X_i}, \val{X_j})$ and $\LCSQ(\val{X_i}, \val{X_j})$ query times can be improved 
to $O(1)$ using $O(n \log n \log N \log^* N)$ preprocessing time. 
\end{theorem}

\begin{theorem}\label{theo:smaller_LCE}
  Given an SLP $\mathcal{S}$ of size $n$ generating a string $T$ of length $N$, 
  there is a data structure 
  which occupies $O(w + n)$ space and 
  supports queries $\LCEQ(\val{X_i},\val{X_j},a,b)$ for
  variables $X_i,X_j$, $1 \leq a \leq |X_i|$ and $1 \leq b
  \leq |X_j|$ in $O(\log N + \log \ell \log^* N)$ time, where $w = O(z \log N \log^* N)$.
  The data structure can be constructed in
  $O(n \log\log n \log N \log^* N)$ preprocessing time and 
  $O(n \log^* N + w)$
  working space, where $z \leq n$ is the size of the LZ77 factorization of $T$ and $\ell$ is the answer of LCE query.
\end{theorem}

Let $h$ be the height of the derivation tree of a given SLP $\mathcal{S}$.
Note that $h \geq \log N$.
Matsubara et al.~\cite{matsubara_tcs2009} showed 
an $O(nh(n + h \log N))$-time $O(n(n + \log N))$-space
algorithm to compute an $O(n \log N)$-size representation of 
all palindromes in the string.
Their algorithm uses
a data structure which supports in $O(h^2)$ time,
$\LCEQ$ queries of a special form $\LCEQ(\val{X_i}, \val{X_j}, 1, p_j)$~\cite{MasamichiCPM97}.
This data structure takes $O(n^2)$ space and can be constructed in 
$O(n^2 h)$ time~\cite{lifshits07:_proces_compr_texts}. 
Using Theorem~\ref{theo:smaller_LCE}, we obtain a faster algorithm,
as follows:
\begin{theorem}\label{theo:palindrome}
Given an SLP of size $n$ generating a string of length $N$,
we can compute an $O(n \log N)$-size representation
of all palindromes in the string 
in $O(n \log^2 N \log^* N)$ time and $O(n \log^* N + w)$ space.
\end{theorem}

Our data structures also solve \emph{the grammar compressed dictionary matching problem}~\cite{DBLP:journals/tcs/INIBT15}.

\begin{theorem}\label{theo:ImproveDictionaryMatching}
Given a DSLP $\langle \mathcal{S}, m\rangle$ of size $n$ 
that represents a dictionary $\Pi_{\langle\mathcal{S},m \rangle}$ for $m$ patterns of total length $N$,
we can preprocess the DSLP in $O((n \log \log n + m \log m) \log N \log^* N)$ time and $O(n \log N \log^* N)$ space
so that, given any text $T$ in a streaming fashion,
we can detect all $\mathit{occ}$ occurrences of the patterns in $T$ in $O(|T|\log m \log N \log^* N + \mathit{occ})$ time.
\end{theorem}

It was shown in~\cite{DBLP:journals/tcs/INIBT15} that we can 
construct in $O(n^4\log n)$ time a data structure of size $O(n^2\log N)$ 
which finds all occurrences of the patterns in $T$ in $O(|T|(h+m))$ time,
where $h$ is the height of
the derivation tree of DSLP $\langle \mathcal{S}, m \rangle$.
Note that our data structure of Theorem~\ref{theo:ImproveDictionaryMatching}
is always smaller, and runs faster when $h = \omega(\log m \log N \log^* N)$.

\section{Appendix: Supplementary Examples and Figures}\label{sec:Example_Section}
\begin{example}[$\encblockd{p}$ and $\encpow{s}$] \label{ex:Encblock}
Let $\log^* W = 2$, and then $\deltaLR{L} = 8, \deltaLR{R} = 4$.\\
If $p = 1,2,3,2,5,7,6,4,3,4,3,4,1,2,3,4,5$ and $d = 1,0,0,1,0,1,0,0,1,0,0,0,1,0,1,0,0$,
then $\encblockd{p} = (1,2,3),(2,5),(7,6,4),(3,4,3,4),(1,2),(3,4,5)$, $|\encblockd{p}| = 6$ and $\encblockd{p}[2] = (2, 5)$. 
For string $s = aabbbbbabb$,
$\encpow{s} = a^2b^5a^1b^2$ and
$|\encpow{s}| = 4$ and $\encpow{s}[2] = b^5$.
\end{example}

\begin{example}[SLP]\label{ex:SLP}
  Let $\mathcal{S} = (\Sigma, \mathcal{V}, \mathcal{D}, S)$ be the SLP
  s.t.
$\Sigma = \{A, B, C \}$, $\mathcal{V} = \{ X_1, \cdots , X_{11} \}$, 
$\mathcal{D} = \{ 
X_{1} \rightarrow A, X_{2} \rightarrow B, X_{3} \rightarrow C, 
X_4 \rightarrow X_{3}X_{1}, X_5 \rightarrow X_{4}X_{2}, 
X_6 \rightarrow X_{5}X_{5}, X_7 \rightarrow X_{2}X_{3}, 
X_8 \rightarrow X_{1}X_{2}, X_9 \rightarrow X_{7}X_{8}, 
X_{10} \rightarrow X_{6}X_{9}, X_{11} \rightarrow X_{10}X_{6}
\}$, $S = X_{11}$, 
the derivation tree of $S$ represents $CABCABBCABCABCAB$.
\end{example}

\begin{example}[RLSLP]\label{ex:tree}
Let $\mathcal{G} = (\Sigma, \mathcal{V}, \mathcal{D}, S)$ be an RLSLP, 
where $\Sigma = \{A, B, C \}$, $\mathcal{V} = \{1, \ldots , 15 \}$, $\mathcal{D} = \{ 
1 \rightarrow A, 2 \rightarrow B, 3 \rightarrow C,
4 \rightarrow 3^4, 5 \rightarrow 1^1, 6 \rightarrow 2^1, 7 \rightarrow 3^1, 
8 \rightarrow (7,5),  9 \rightarrow (8,6), 10 \rightarrow (5,6), 11 \rightarrow (10,4), 
12 \rightarrow 9^2, 13 \rightarrow 10^7, 14 \rightarrow 11^1, 15 \rightarrow (12,13), 16 \rightarrow (15,14), 
17 \rightarrow 16^1
\}$, and $S = 17$.
The derivation tree of the start symbol $S$ represents a single string 
$T = CABCABABABABABABABABABCCCC$. 
Here, $\sig{(7,5)} = 8$, $\sig{(7,5,6)} = 9$, $\sig{(6,5)} = \rm{undefined}$.
See also Fig.~\ref{fig:SignatureTree}
which illustrates the derivation tree of the start symbol $S$
and the DAG for $\mathcal{G}$.
\end{example}

\begin{example}[Signature encoding]\label{ex:signature_dictionary}
Let $\mathcal{G} = (\Sigma, \mathcal{V}, \mathcal{D}, S)$ be an RLSLP of Example~\ref{ex:tree}. 
Assuming $\encblock{\pow{0}{T}} = (7,5,6),(7,5,6),(5,6)^7,(5,6,4)$ and $\encblock{\pow{1}{T}} = (12,13,14)$ hold, 
$\mathcal{G}$ is the signature encoding of $T$ and $\id{T} =  17$.
See Fig.~\ref{fig:SignatureTree} for an illustration of the derivation tree of $\mathcal{G}$ and the corresponding DAG.
\end{example}

\begin{figure}[ht]
\begin{center}
  \includegraphics[scale=0.7]{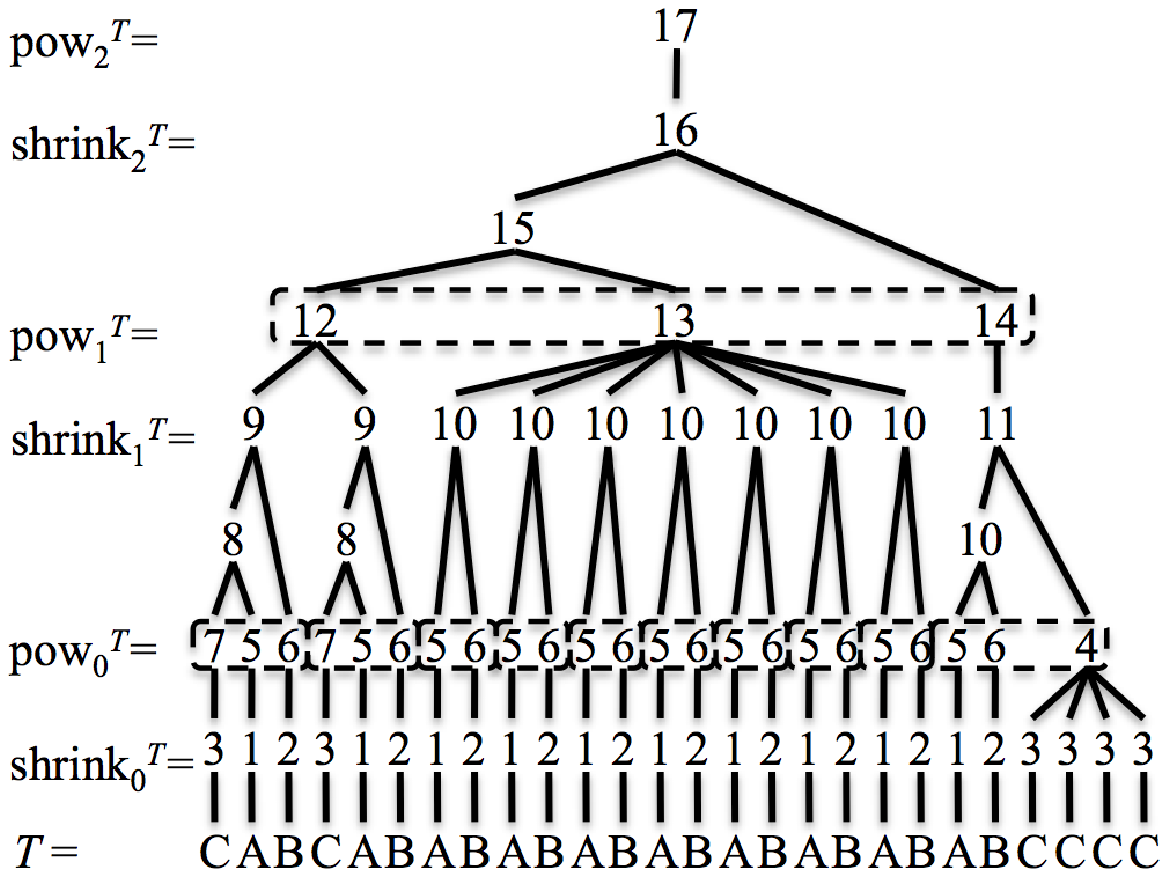}  
  \includegraphics[scale=0.7]{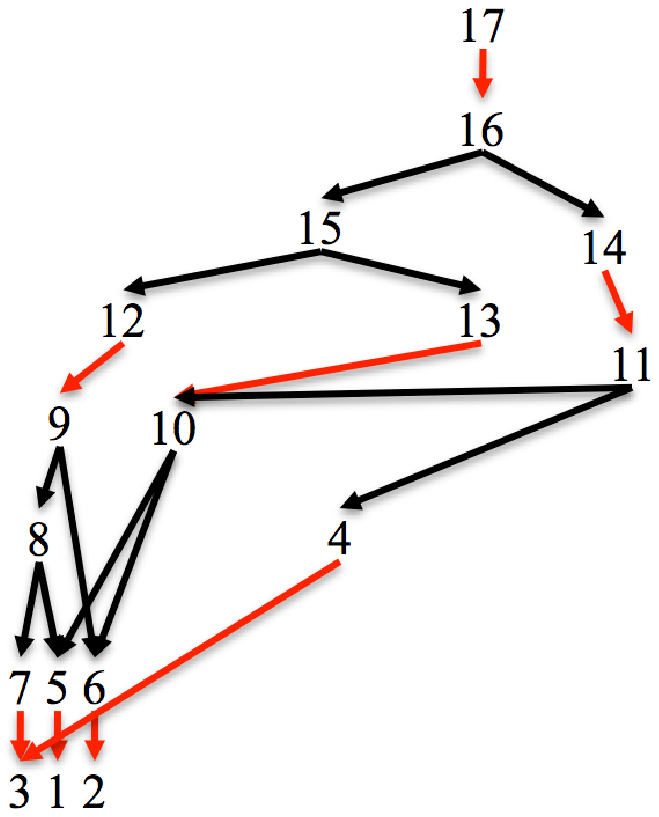}
  \caption{
  The derivation tree of $S$ (left) and the DAG for $\mathcal{G}$ (right) of Example~\ref{ex:tree}. 
  In the DAG, the black and red arrows represent $e \rightarrow e_{\ell}e_r$ and
  $e \rightarrow \hat{e}^{k}$ respectively.   
  In Example~\ref{ex:signature_dictionary}, $T$ is encoded by signature encoding.
  In the derivation tree of $S$, the dotted boxes represent the blocks created by the $\mathit{Eblock}$ function.
  }
  \label{fig:SignatureTree}
\end{center}
\end{figure}

\begin{figure}[h]
\begin{center}
  \includegraphics[scale=0.5]{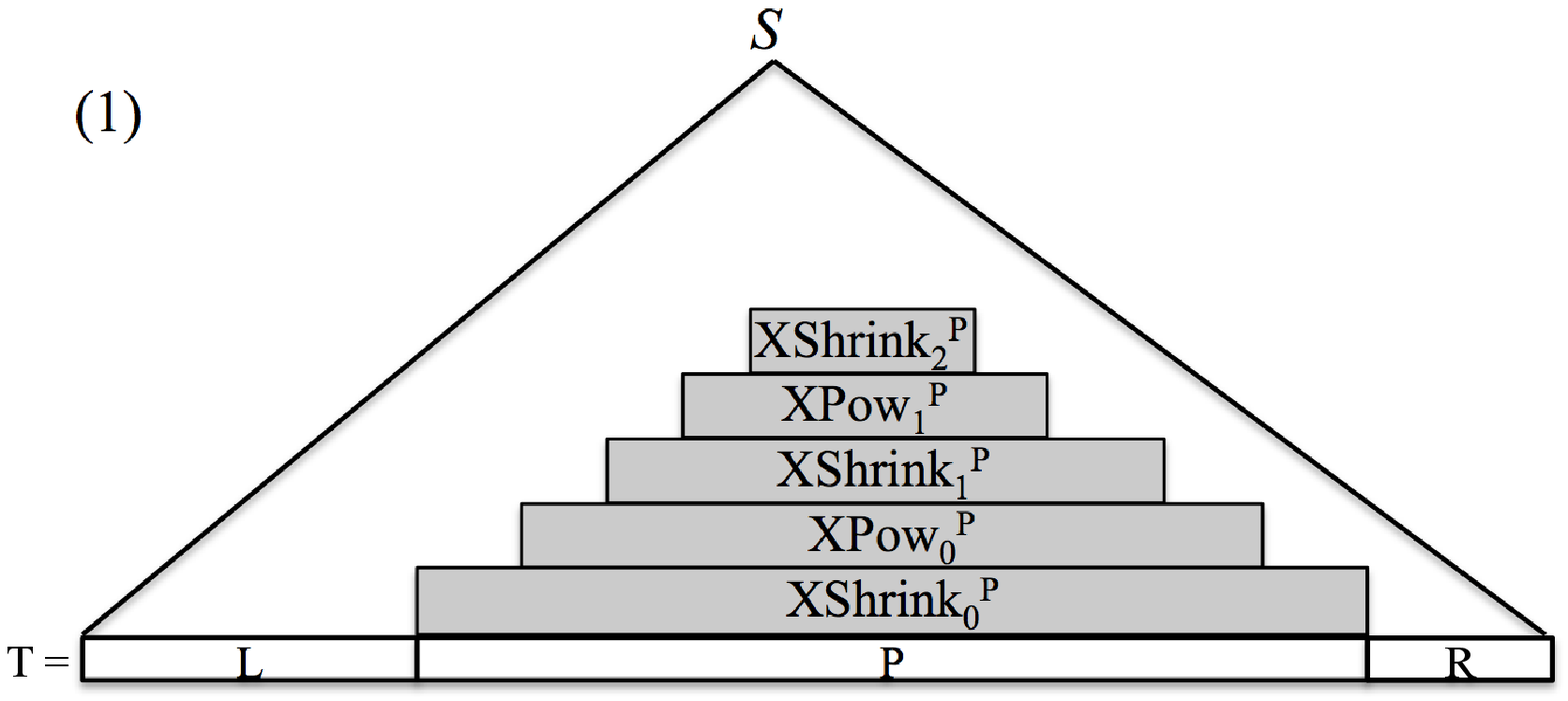}
  \includegraphics[scale=0.6]{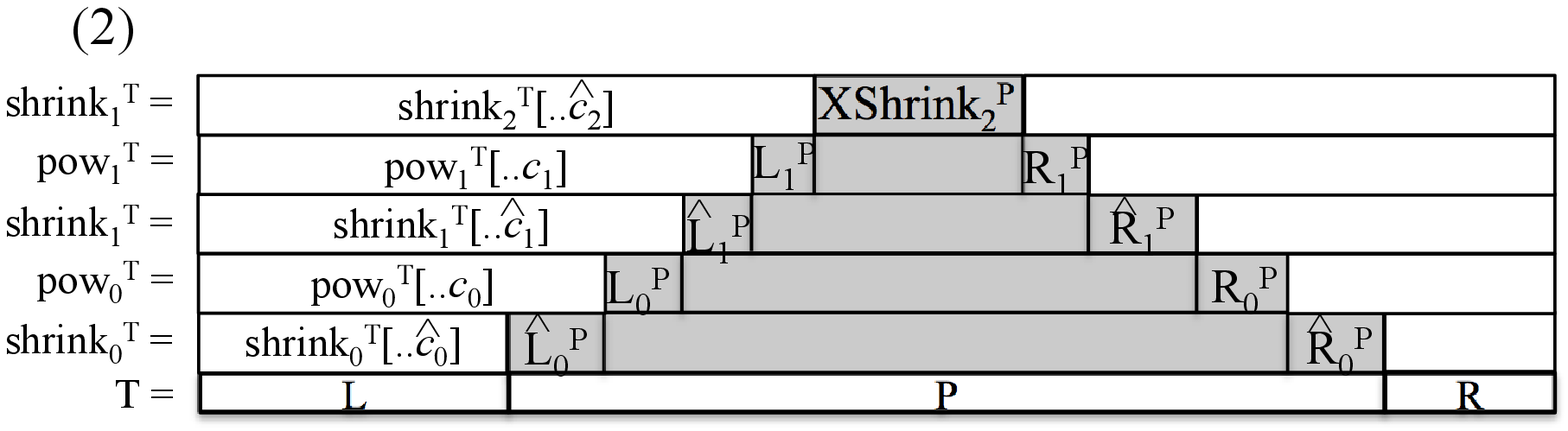}
  \includegraphics[scale=0.5]{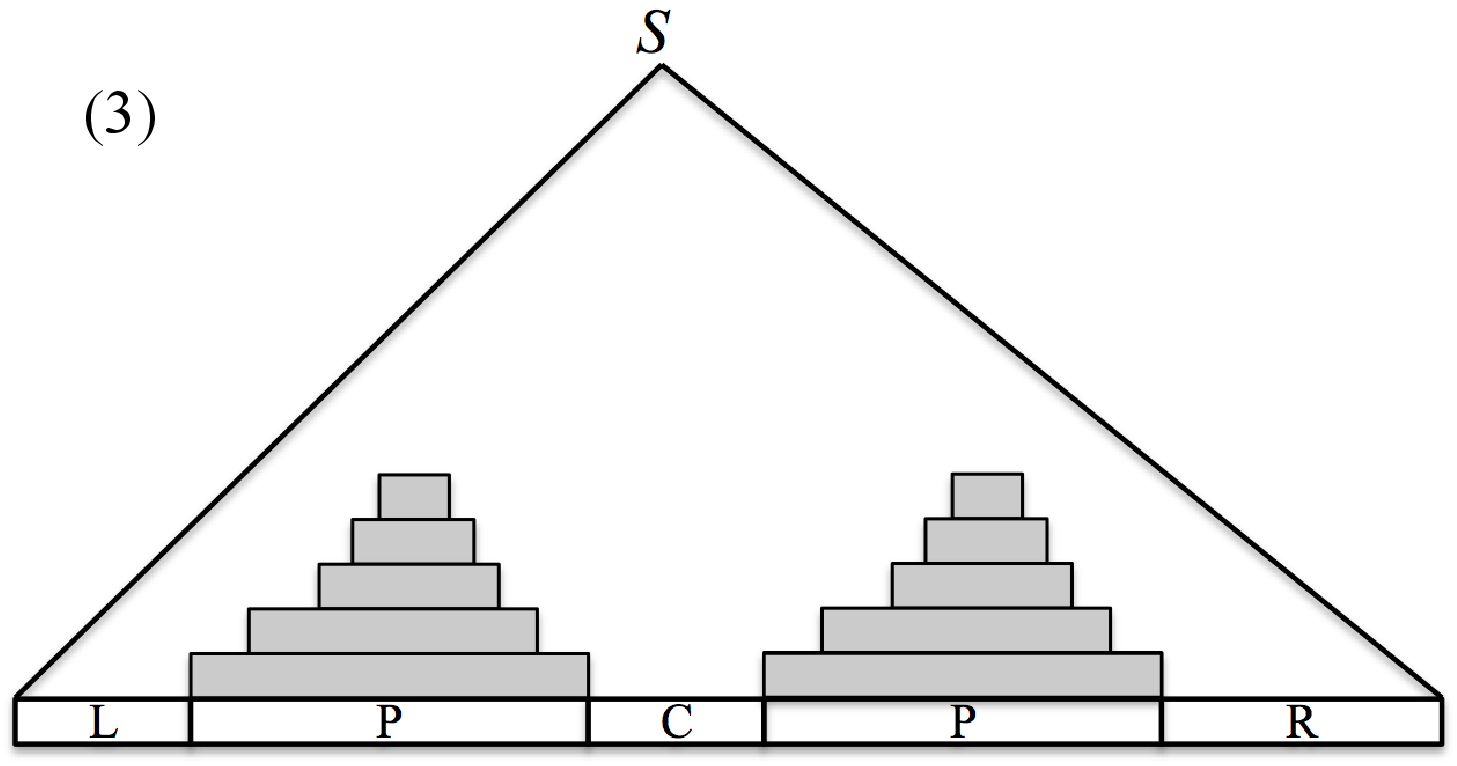}
  \caption{
  Abstract images of consistent signatures of substring $P$ of text $T$,
  on the derivation trees of the signature encoding of $T$.
  Gray rectangles in Figures (1)-(3) represent common signatures for occurrences of $P$. 
  (1) Each $\mathit{XShrink}_{t}^{P}$ and $\mathit{XPow}_{t}^{P}$ occur on substring $P$ 
  in $\mathit{shrink}_{t}^{T}$ and $\mathit{Pow}_{t}^{T}$, respectively, where $T = LPR$.
  (2) The substring $P$ can be represented by 
  $\hat{L}_{0}^{P}L_{0}^{P}\hat{L}_{1}^{P}L_{1}^{P}\mathit{XShrink}_{2}^{P}R_{1}^{P}\hat{R}_{1}^{P}R_{0}^{P}\hat{R}_{0}^{P}$. 
  (3) There exist common signatures on every substring $P$ in the derivation tree.
  } 
  \label{fig:CommonSequence}
\end{center}
\end{figure}

\section{Appendix: Proof of Lemma~\ref{lem:CoinTossing}}\label{sec:proof_prelim}
\begin{proof}
Here we give only an intuitive description of a proof of Lemma~\ref{lem:CoinTossing}.
More detailed proofs can be found at~\cite{DBLP:journals/algorithmica/MehlhornSU97}
and~\cite{LongAlstrup}.

Let $p$ be an integer sequence of length $n$,
called a \emph{$W$-colored sequence},
where $p[i] \neq p[i+1]$ for any $1 \leq i < n$ and $0 \leq p[j] \leq W$ for any $1 \leq j \leq n$.
Mehlhorn et al.~\cite{DBLP:journals/algorithmica/MehlhornSU97} showed that there exists a function $f'$  
which returns a $(\log W)$-colored sequence $p'$ for a given $W$-colored sequence $p$ in $O(|p|)$ time, 
where $p'[i]$ is determined only by $p[i-1]$ and $p[i]$ for $1 \leq i \leq |p|$. 
Let $p^{\langle k \rangle}$ denote the outputs after applying $f'$ to $p$ by $k$ times.
They also showed that there exists a function $f''$
which returns a bit sequence $d$ satisfying the conditions of Lemma~\ref{lem:CoinTossing} 
for a $6$-colored sequence $p$ in $O(|p|)$ time, 
where $d[i]$ is determined only by $p[i-3..i+3]$ for $1 \leq i \leq |p|$. 
Hence we can compute $d$ for a $W$-colored sequence $p$ 
in $O(|p| \log^* W)$ time by applying $f''$ to $p^{\langle \log^* W + 2 \rangle}$ after computing $p^{\langle \log^* W + 2 \rangle}$. 
Furthermore, Alstrup et al.~\cite{LongAlstrup} showed
that $d$ can be computed in $O(|p|)$ time using a precomputed table of size $o(\log W)$.
The idea is that $p^{\langle 3 \rangle}$ is a $\log\log\log W$-colored sequence and the number of 
all combinations of a $\log\log\log W$-colored sequence of length $\log^* W + 11$ is $2^{(\log^* W + 11)\log\log\log W} = o(\log W)$. 
Hence we can compute $d$ for a $W$-colored sequence in linear time using a precomputed table of size $o(\log W)$.
\end{proof}

\section{Appendix: Omitted Proofs in Sections~\ref{sec:lce} and~\ref{sec:Update}}\label{sec:appendix_proof_sig}
\subsection{Proof of Lemma~\ref{lem:common_sequence2}}
\begin{proof}
	Consider any integer $i$ with $T[i..i+|P|-1] = P$
	(see also Fig.~\ref{fig:CommonSequence}(2)). 
	Note that for $0 \leq t < h^{P}$,
	if $\mathit{XShrink}_{t}^{P}$ occurs in $\shrink{t}{T}$, then 
	$\mathit{XPow}_{t}^{P}$ always occurs in $\pow{t}{T}$,
	because $\mathit{XPow}_{t}^{P}$ is determined only by $\mathit{XShrink}_{t}^{P}$. 
	Similarly, for $0 < t \leq h^{P}$, if $\mathit{XPow}_{t-1}^{P}$ occurs in $\pow{t-1}{T}$, 
	then $\mathit{XShrink}_{t}^{P}$ always occurs in $\shrink{t}{T}$.
	Since $\mathit{XShrink}_{0}^{P}$ occurs at position $i$ in $\shrink{0}{T}$, 
	$\mathit{XShrink}_{t}^{P}$ and $\mathit{XPow}_{t}^{P}$ occur in the derivation tree of $\id{T}$. 
	Hence we discuss the positions of $\mathit{XShrink}_{t}^{P}$ and $\mathit{XPow}_{t}^{P}$.
	Now, let $\hat{c}_{t}$ + 1 and $c_{t}$ + 1 be the beginning positions of
	the corresponding occurrence of $\mathit{XShrink}_{t}^{P}$ in $\shrink{t}{T}$ and 
	that of $\mathit{XPow}_{t}^{P}$ in $\pow{t}{T}$, respectively. 
	Then $\shrink{t}{T}[..\hat{c}_{t}]$ consists of $\pow{t-1}{T}[..c_{t-1}]$ and $L_{t-1}^{P}$ for $0 < t \leq h^{P}$. 
	Also, $\pow{t}{T}[..c_{t}]$ consists of $\shrink{t}{T}[..\hat{c}_{t}]$ and $\hat{L}_{t}^{P}$ for $0 \leq t < h^{P}$. 
		This means that 
		the substring $P$ occurring at position $i$ in $T$ is  
		represented as $\uniq{P}$ in the signature encoding
	Therefore Lemma~\ref{lem:common_sequence2} holds.
\end{proof}

\subsection{Proof of Lemma~\ref{lem:ancestors}}
\begin{proof}
  By Definition~\ref{def:xshrink}, for every level,
  $X$ contains $O(\log^* M)$ nodes that are parents of the nodes representing $\uniq{P}$.
  Lemma~\ref{lem:ancestors} holds because the number of nodes at some level is halved when $\mathit{Shrink}$ is applied.   	
  More precisely, considering the $x$ nodes of $X$ at some level to which $\mathit{Shrink}$ is applied,
  the number of their parents is at most $(x + 2) / 2$. 
  	Here the `+2' term reflects the fact that both ends of $x$ nodes may be coupled with nodes outside $X$.
  	And also, since $|\encpow{\hat{L}_{t}^{P}}| = |\encpow{\hat{R}_{t}^{P}}| = 1$ for $0 \leq t < h^{P}$ and $|\encpow{\xshrink{h^{P}}{P}}| = O(|\log^* M|)$, 
  	each nodes representing $\hat{L}_{t}^{P}$ and $\hat{R}_{t}^{P}$ has a common parent for every level, 
  	and the number of parents of nodes representing $\xshrink{h^{P}}{P}$ is $O(\log^* M)$.
  Note that $h = O(\log |\val{e}|)$ holds for $e \in \mathcal{V}$ 
  by the signature encoding, where $h$ is the height of derivation tree of $e$.    
\end{proof}

\subsection{Proof of Lemma~\ref{lem:ComputeShortCommonSequence}}
\begin{proof}
  	Let $\mathcal{T}$ be the derivation tree of $e$ and 
  	consider the induced subtree $X$ of $\mathcal{T}$
  	whose root is the root of $\mathcal{T}$ and whose leaves are 
  	the parents of the nodes representing $\uniq{s[j..j+y-1]}$.
  	Then the size of $X$ is $O(\log |s| + \log y \log^* M)$ by Lemma~\ref{lem:ancestors}. 
  	Starting at the given node in the DAG which corresponds to $e$, 
  	we compute $X$ using Definition~\ref{def:xshrink} and the properties described in the proof of Lemma~\ref{lem:ancestors} 
  	in $O(\log |s| + \log y \log^* M)$ time. 
  	Hence Lemma~\ref{lem:ComputeShortCommonSequence} holds.
\end{proof}

\section{Appendix: Omitted Proofs in Section~\ref{sec:Construction}}\label{sec:Proof_HConstructuionTheorem}

\subsection{Proof of Theorem~\ref{theo:HConstructuionTheorem}~(2)}
\begin{proof}
	Consider a dynamic signature encoding $\mathcal{G}$ for an empty string. 
	Then Theorem~\ref{theo:HConstructuionTheorem}~(2) immediately holds 
	by computing $\mathit{INSERT'}(c_i,|f_i|,|f_1 \cdots f_{i-1}|+1)$ for all $1 \leq i \leq z$ incrementally, 
	where $c_i \leq |f_1 \cdots f_{i-1}| - |f_i|$ is a position such that $T[c_i..c_i+|f_i|-1] = f_i$ holds.
		Note that when $f_i$ is a character which does not occur in $f_1, \ldots f_{i-1}$ for $1 \leq i \leq z$, 
		we compute $\mathit{INSERT}(f_i,|f_1 \cdots f_{i-1}|+1)$ in $O(f_{\mathcal{A}} \log N \log^* M)$ time
		instead of the above $\mathit{INSERT'}$ operation.
\end{proof}
Note that we can directly show Lemma~\ref{lem:upperbound_signature} from the above proof 
because the size of $\mathcal{G}$ increases $O(\log N \log^* M)$ by Lemma~\ref{lem:ancestors}, 
every time we do $\mathit{INSERT'}(c_i,|f_i|,|f_1 \cdots f_{i-1}|+1)$ for $1 \leq i \leq z$.

\subsection{Proof of Theorem~\ref{theo:HConstructuionTheorem}~(3a)}
\begin{proof}
	We use the \emph{G-factorization} proposed in~\cite{rytter03:_applic_lempel_ziv}. 
	By the G-factorization of $T$ with respect to $\mathcal{S}$,
 	$T$ is partitioned into $O(n)$ strings, each of which, corresponding to $T[i..j]$, is derived by a variable $X$ of $\mathcal{S}$
	such that $X$ appears in the derivation tree of $\mathcal{S}$ to derive a substring of $T[1..i-1]$, 
	or otherwise $X$ derives a single character that does not appear in $T[1..i-1]$.
	Note that we can compute a sequence of variables of $\mathcal{S}$ 
	corresponding to the G-factorization of $T$ with respect to $\mathcal{S}$
 	in $O(n)$ time by the depth-first traversal of the DAG of $S$.
	Since the G-factorization resembles the LZ77 factorization, 
	we can construct the dynamic signature encoding $\mathcal{G}$ for $T$ by 
	$O(n)$ $\mathit{INSERT'}$ and $\mathit{INSERT}$ operations 
	as the proof of Theorem~\ref{theo:HConstructuionTheorem}~(2).
\end{proof}

\subsection{Proof of Lemma~\ref{lem:Lambda_t}}\label{sec:Proof_Lambda_t}
\begin{proof}
We first compute, for all variables $X_i$,
$\encpow{\xshrink{t}{X_i}}$ if $|\encpow{\xshrink{t}{X_i}}| \leq \Delta_{L} + \Delta_{R} + 9$,
otherwise $\encpow{\hat{L}_{t}^{X_i}}$ and $\encpow{\hat{R}_{t}^{X_i}}$.
The information can be computed in $O(n \log^*M)$ time and space in a bottom-up manner, i.e., by processing variables in increasing order.
For $X_i \rightarrow X_{\ell} X_{r}$, if both $|\encpow{\xshrink{t}{X_{\ell}}}|$ and $|\encpow{\xshrink{t}{X_{r}}}|$ are no greater than $\Delta_{L} + \Delta_{R} + 9$,
we can compute $\encpow{\xshrink{t}{X_i}}$ in $O(\log^* M)$ time 
by naively concatenating $\encpow{\xshrink{t}{X_{\ell}}}$, $\encpow{\hat{z}_{t}^{X_i}}$ and $\encpow{\xshrink{t}{X_{r}}}$.
Otherwise $|\encpow{\xshrink{t}{X_i}}| > \Delta_{L} + \Delta_{R} + 9$ must hold, and 
$\encpow{\hat{L}_{0}^{X_i}}$ and $\encpow{\hat{R}_{0}^{X_i}}$ can be computed in $O(1)$ time from $\encpow{\hat{z}_{t}^{X_i}}$ and the information for $X_{\ell}$ and $X_{r}$.

The run-length encoded signatures represented by $z_{t}^{X_i}$ can be obtained in $O(\log^* M)$ time
by using $\hat{z}_{t}^{X_i}$ and the above information for $X_{\ell}$ and $X_r$:
$z_{t}^{X_i}$ is created over run-length encoded signatures that are obtained by concatenating 
$\encpow{\xshrink{0}{X_{\ell}}}$ (or $\encpow{\hat{R}_{0}^{X_{\ell}}}$), $z_{t}^{X_i}$ and $\encpow{\xshrink{0}{X_r}}$ (or $\encpow{\hat{R}_{0}^{X_r}}$).
Also, $A_{t}^{X_n}$ and $B_{t}^{X_n}$ represents $\hat{A}_{t}^{X_n} \hat{L}_{t}^{X_n}$ and $\hat{R}_{t}^{X_n} \hat{B}_{t}^{X_n}$, respectively.

Hence, we can compute in $O(n \log^* M)$ time $O(n \log^*M)$ run-length encoded signatures to which we give signatures.
We determine signatures in $O(n \log \log (n \log^* M) \log^* M)$ time by sorting the run-length encoded signatures as Lemma~\ref{lem:Lambda_t}.
\end{proof}

\section*{Appendix D: Omitted Proofs in Section~\ref{sec:applications}}\label{sec:appendix_applications}

\subsection{Proof of Theorem~\ref{theo:changedSLP}}

\subsubsection{Proof of Theorem~\ref{theo:changedSLP}~(1)}
\begin{proof}
For any signature $e \in \mathcal{V}$ such that $e \rightarrow e_{\ell}e_r$, 
we can easily translate $e$ to a production of SLPs because the assignment is a pair of signatures, 
like the right-hand side of the production rules of SLPs.
For any signature $e \in \mathcal{V}$ such that $e \rightarrow \hat{e}^k$, 
we can translate $e$ to at most $2 \log k$ production rules of SLPs:
We create $t = \lfloor \log k \rfloor$ variables which represent
$\hat{e}^{2^1}, \hat{e}^{2^2}, \ldots, \hat{e}^{2^t}$ and concatenating them
according to the binary representation of $k$ to make up $k$ $\hat{e}$'s.
Therefore we can compute $\mathcal{S}$ in $O(w \log |T|)$ time.
\end{proof}

\subsubsection{Proof of Theorem~\ref{theo:changedSLP}~(2)}
\begin{proof}
Note that the number of created or removed signatures in $\mathcal{V}$ 
is bounded by $O(y + \log |T'| \log^* M)$ by Lemma~\ref{lem:ancestors}. 
For each of the removed signatures, we remove the corresponding production from $\mathcal{S}$. 
For each of created signatures, we create the corresponding production and add it to $\mathcal{S}$ as in the proof of (1). 
Therefore Theorem~\ref{theo:changedSLP} holds. 
\end{proof}

\subsection{Proof of Theorem~\ref{theo:faster_LCP}}
We use the following known result.
\begin{lemma}[\cite{LongAlstrup}]\label{lem:lcp_lcs_on_H}
Using signature encodings $\mathcal{G}_{1}, \ldots \mathcal{G}_{m}$, 
we can support 
\begin{itemize}
 \item $\mathit{LCP}(T_i, T_j)$ in $O(\log |T_i| + \log |T_j|)$ time,
 \item $\mathit{LCS}(T_i, T_j)$ in $O((\log |T_i| + \log |T_j|) \log^* M)$ time
\end{itemize}
where $T_i, T_j \in \{ T_1, \ldots, T_m \}$ and 
$\mathcal{G}_{k} =(\Sigma, \mathcal{V}, \mathcal{D}, S_{k})$ of a string $T_k$ for $1 \leq k \leq m$, 
namely $\mathcal{G}_{1}, \ldots, \mathcal{G}_{m}$ share $\mathcal{D}$.
\end{lemma}
\begin{proof}

We compute $\mathit{LCP}(T_i, T_j)$ by $\mathit{LCE}(T_i,T_j,1,1)$,  
namely, we use the algorithm of Lemma~\ref{lem:sub_operation_lemma}. 
Let $P$ denote the longest common prefix of $T_i$ and $T_j$.
We use the notation $\hat{A}^{P}$ defined in Section~\ref{sec:HConstruction3-2}.
 	Then the both substrings $P$ occurring at position $1$ 
 	in $T_i$ and at position $1$ in $T_j$ are represented as 
 	$v = \hat{A}^{P}_{h^{P}}\xshrink{h^{P}}{P}R_{h^{P}-1}^{P}\hat{R}_{h^{P}-1}^{P} \cdots R_{0}^{P}\hat{R}_{0}^{P}$ 
 	in the signature encoding by a similar argument of Lemma~\ref{lem:common_sequence2}. 
Since $|\encpow{v}| = O(\log |P| + \log^* M)$, 
we can compute $\mathit{LCP}(T_i, T_j)$ in $O(\log |T_i| + \log |T_j|)$ time. 
Similarly, we can compute $\mathit{LCS}(T_i, T_j)$ in $O((\log |T_i| + \log |T_j|) \log^* M)$ time.
More detailed proofs can be found in~\cite{LongAlstrup}.
\end{proof}
To use Lemma~\ref{lem:lcp_lcs_on_H} for $\id{\val{X_1}}, \ldots, \id{\val{X_n}}$, we show the following lemma. 
\begin{lemma}\label{lem:SLPSignatureEncoding}
Given an SLP $\mathcal{S}$, 
we can compute $\id{\val{X_1}}, \ldots, \id{\val{X_n}}$ in \\
$O(n \log\log n \log N \log^* M)$ time and $O(n \log N \log^* M)$ space. 
\end{lemma}
\begin{proof}
Recall that the algorithm of Theorem~\ref{theo:HConstructuionTheorem}~(3) 
computes $\id{\val{X_n}}$ in $O(n \log\log n \log N \log^* M)$ time.
We can modify the algorithm to compute $\id{\val{X_1}}, \ldots, \id{\val{X_n}}$ without changing the time complexity:
We just compute $A_{t}^{X}$, $\hat{A}_{t}^{X}$, $B_{t}^{X}$ and $\hat{B}_{t}^{X}$ for ``all'' $X \in \mathcal{S}$, not only for $X_n$.
Since the total size is $O(n \log N \log^* M)$, Lemma~\ref{lem:SLPSignatureEncoding} holds.
\end{proof}

We are ready to prove Theorem~\ref{theo:faster_LCP}.
\begin{proof}
The first result immediately follows from Lemma~\ref{lem:lcp_lcs_on_H} and~\ref{lem:SLPSignatureEncoding}. 
To speed-up query times for $\LCPQ$ and $\LCSQ$,
we sort variables in lexicographical order 
in $O(n \log n \log N)$ time by $\LCPQ$ query and a standard comparison-based sorting.
Constant-time $\LCPQ$ queries are then possible by 
using a constant-time RMQ data structure~\cite{DBLP:journals/jal/BenderFPSS05}
on the sequence of the lcp values. 
Next we show that $\LCSQ$ queries can be supported similarly.  
Let SLP $\mathcal{S} = (\Sigma, \mathcal{V}, \mathcal{D},S)$ and 
$Y_i \rightarrow \mathit{expr}_{i}$ for $1 \leq i \leq n$, where
$\mathit{expr}_{i} = Y_{r}Y_{\ell}$ for $X_{i} \rightarrow X_{\ell}X_{r} \in \mathcal{D}$ and 
$\mathit{expr}_{i} = a$ for $(X_{i} \rightarrow a \in \Sigma) \in \mathcal{D}$. 
Then consider an SLP $\mathcal{S'} = (\Sigma, \mathcal{V'}, \mathcal{D},S')$ of size $n$, where 
$\mathcal{V'} = \{ Y_{1}, \ldots, Y_{n} \}$, $\mathcal{D'} = \{ Y_{1} \rightarrow \mathit{expr}_{i}, \ldots, Y_{n} \rightarrow \mathit{expr}_{n} \}$ and 
$S' = Y_{n}$. 
Namely $\mathcal{S'}$ represents $T^{R}$. 
By supporting $\LCPQ$ queries on $\mathcal{S'}$, 
$\LCSQ$ queries on $\mathcal{S}$ can be supported. 
Hence Theorem~\ref{theo:faster_LCP} holds. 
\end{proof}

\subsection{Proof of Theorem~\ref{theo:smaller_LCE}}
\begin{proof}
We can compute a static signature encoding $\mathcal{G} = (\Sigma, \mathcal{V}, \mathcal{D}, S)$ of size $w$ 
representing $T$ in $O(n \log \log n \log N \log^* M)$ time 
and $O(n \log^* M + w)$ working space using Theorem~\ref{theo:HConstructuionTheorem}, 
where $w = O(z \log N \log^* M)$.
Notice that 
each variable of the SLP appears at least once in the 
derivation tree of $T_{n}$ of the last variable $X_n$ representing the string $T$.
Hence, if we store an occurrence of each variable $X_i$ in $\mathcal{T}_n$
and $|\val{X_i}|$, we can reduce any LCE query on two variables 
to an LCE query on two positions of $\val{X_n} = T$.
In so doing, for all $1 \leq i \leq n$  
we first compute $|\val{X_i}|$ 
and then compute the leftmost occurrence $\ell_i$ of $X_i$ in $\mathcal{T}_n$,
spending $O(n)$ total time and space.
By Lemma~\ref{lem:sub_operation_lemma}, 
each LCE query can be supported in $O(\log N + \log \ell \log^* M)$ time.
Since $z \leq n$~\cite{rytter03:_applic_lempel_ziv},
the total preprocessing time is $O(n \log \log n \log N \log^* M)$
and working space is $O(n \log^* M + w)$.
\end{proof}
\subsection{Proof of Theorem~\ref{theo:palindrome}}

\begin{proof}
For a given SLP of size $n$ representing a string of length $N$,
let $P(n,N)$, $S(n, N)$, and $E(n,N)$ be
the preprocessing time and space requirement for an $\LCEQ$ data structure on
SLP variables, and each $\LCEQ$ query time, respectively.

Matsubara et al.~\cite{matsubara_tcs2009} showed that 
we can compute an $O(n \log N)$-size representation
of all palindromes in the string
in $O(P(n,N) + E(n,N) \cdot n \log N)$ time and $O(n \log N + S(n, N))$ space.
Hence, using Theorem~\ref{theo:smaller_LCE},
we can find all palindromes in the string in 
$O(n \log \log n \log N \log^* M + n \log^2 N \log^* M) = O(n \log^2 N \log^* M)$ time
and $O(n \log^* M + w)$ space.
\end{proof}

\subsection{Proof of Theorem \ref{theo:ImproveDictionaryMatching}}

\begin{proof}
In the preprocessing phase, 
we construct a static signature encoding $\mathcal{G} = (\Sigma, \mathcal{V}, \mathcal{D}, S)$ of size $w'$ 
such that $\id{\val{X_{1}}}, \ldots ,\id{\val{X_n}} \in \mathcal{V}$ 
using Lemma~\ref{lem:SLPSignatureEncoding}, spending $O(n \log \log n \\ \log N \log^* M)$ 
time, where $w' = O(n \log N \log^* M)$. 
Next we construct a compacted trie of size $O(m)$ that represents the $m$ patterns, which we denote by \emph{$\mathit{PTree}$ (pattern tree)}.
Formally, each non-root node of $\mathit{PTree}$ represents either a pattern or 
the longest common prefix of some pair of patterns.
$\mathit{PTree}$ can be built by using $\LCPQ$ of Theorem~\ref{theo:faster_LCP} in $O(m \log m \log N)$ time.
We let each node have its string depth, and the pointer to its deepest ancestor node that represents a pattern if such exists.
Further, we augment $\mathit{PTree}$ with a data structure for level ancestor queries so that 
we can locate any prefix of any pattern, designated by a pattern and length, in $\mathit{PTree}$ in $O(\log m)$ time
by locating the string depth by binary search on the path from the root to the node representing the pattern.
Supposing that we know the longest prefix of $T[i..|T|]$ that is also a prefix of one of the patterns, which we call the \emph{max-prefix for $i$},
$\mathit{PTree}$ allows us to output $\mathit{occ}_i$ patterns occurring at position $i$ in $O(\log m + \mathit{occ}_i)$ time.
Hence, the pattern matching problem reduces to computing the max-prefix for every position.

In the pattern matching phase, our algorithm processes $T$ in a streaming fashion, i.e.,
each character is processed in increasing order and discarded before taking the next character.
Just before processing $T[j+1]$, 
the algorithm maintains a pair of signature $p$ and integer $l$
such that $\derive(p)[1..l]$ is the longest suffix of $T[1..j]$ that is also a prefix of one of the patterns.
When $T[j+1]$ comes, we search for the smallest position $i \in \{j-l+1, \dots, j+1\}$ such that there is a pattern whose prefix is $T[i..j+1]$.
For each $i \in \{j-l+1, \dots, j+1\}$ in increasing order,
we check if there exists a pattern whose prefix is $T[i..j+1]$ by binary search on a sorted list of $m$ patterns.
Since $T[i..j] = \derive(p)[i-j+l..l]$, $\LCEQ$ with $p$ can be used for comparing a pattern prefix and $T[i..j+1]$ (except for the last character $T[j+1]$),
and hence, the binary search is conducted in $O(\log m \log N \log^* M)$ time.
For each $i$, if there is no pattern whose prefix is $T[i..j+1]$,
we actually have computed the max-prefix for $i$, and then we output the occurrences of patterns at $i$.
The time complexity is dominated by the binary search, which takes place $O(|T|)$ times in total.
Therefore, the algorithm runs in $O(|T|\log m \log N \log^* M + \mathit{occ})$ time.

By the way, one might want to know occurrences of patterns as soon as they appear
as Aho-Corasick automata do it by reporting the occurrences of the patterns by their ending positions.
Our algorithm described above can be modified to support it without changing the time and space complexities.
In the preprocessing phase, we additionally compute \emph{$\mathit{RPTree}$ (reversed pattern tree)},
which is analogue to $\mathit{PTree}$ but defined on the reversed strings of the patterns,
i.e., $\mathit{RPTree}$ is the compacted trie of size $O(m)$ that represents the reversed strings of the $m$ patterns.
Let $T[i..j]$ be the longest suffix of $T[1..j]$ that is also a prefix of one of the patterns.
A suffix $T[i'..j]$ of $T[i..j]$ is called the \emph{max-suffix for $j$} iff
it is the longest suffix of $T[i..j]$ that is also a suffix of one of the patterns.
Supposing that we know the max-suffix for $j$,
$\mathit{RPTree}$ allows us to output $\mathit{eocc}_j$ patterns occurring with ending position $j$ in $O(\log m + \mathit{eocc}_j)$ time.
Given a pair of signature $p$ and integer $l$ such that $T[i..j] = \derive(p)[1..l]$,
the max-suffix for $j$ can be computed in $O(\log m \log N \log^* M)$ time by binary search on a list of $m$ patterns sorted by their ``reversed'' strings
since each comparison can be done by ``leftward'' $\LCEQ$ with $p$.
Except that we compute the max-suffix for every position and output the patterns ending at each position,
everything else is the same as the previous algorithm, and hence, the time and space complexities are not changed.
\end{proof}

\vspace*{1pc}

\clearpage
\bibliography{ref}

\end{document}